\documentclass[11pt]{article}
\usepackage{fullpage}
\pdfoutput=1
\let\oldmarginpar\marginpar
\renewcommand\marginpar[1]{\-\oldmarginpar[\raggedleft\footnotesize #1]%
{\raggedright\footnotesize #1}}

\usepackage{graphicx, alltt, xspace, cite}
\usepackage{verbatim, url}
\usepackage{comment}
\usepackage{multirow}
\usepackage{float}
\usepackage{algorithmic}
\usepackage[ruled,vlined]{algorithm2e}
\usepackage{array}
\usepackage{amsmath,amsthm,amsfonts,amsgen,amstext,dsfont,amssymb}

\usepackage{color}

\newcommand{\fullversion}[1]{#1}
\newcommand{\submissionversion}[1]{}

\def\compactify{\itemsep=0pt \topsep=0pt \partopsep=0pt \parsep=0pt}
\let\latexusecounter=\usecounter

\def\compactifytwo{\itemsep=-1pt \topsep=-1pt \partopsep=-2pt \parsep=-1pt \labelwidth=-2pt \leftmargin=-10pt}

\newtheorem{theorem}{Theorem}[section]
\newtheorem{proposition}{Proposition}[section]
\newtheorem{lemma}[theorem]{Lemma}      % Corollary environment.

\newcommand{\KL}{\ensuremath \mathsf{KL}}
\newcommand{\NP}{\textsf{NP}}

\newcommand{\set}[1]{\left\{ #1 \right\}}
\newcommand{\setof}[2]{\{{#1}\mid{#2}\}}        % Set (as in \setof{x}{x>0}).
\newcommand{\norm}[1]{\left\| #1 \right\|}
\newcommand{\size}[1]{\left| #1 \right|}

\newcommand{\eat}[1]{}
\begin{document}

\title{Probabilistic Management of OCR Data using an RDBMS}

\author{
Arun Kumar \\
University of Wisconsin-Madison\\
arun@cs.wisc.edu 
\and
Christopher R\'{e}\\
University of Wisconsin-Madison\\
chrisre@cs.wisc.edu
}

\newcounter{example}[section]
\newenvironment{example}{\refstepcounter{example}\noindent\textbf{Example \arabic{example}}.}{$\qed$}

\newcommand{\name}{\textsc{Staccato}\xspace}

\maketitle

\begin{abstract}
The digitization of scanned forms and documents is changing the data
sources that enterprises manage. To integrate these new data sources
with enterprise data, the current state-of-the-art approach is to
convert the images to ASCII text using optical character recognition
(OCR) software and then to store the resulting ASCII text in a
relational database. The OCR problem is challenging, and so the output
of OCR often contains errors. In turn, queries on the output of OCR
may fail to retrieve relevant answers. State-of-the-art OCR programs,
e.g., the OCR powering Google Books, use a probabilistic model that
captures many alternatives during the OCR process. Only when the
results of OCR are stored in the database, do these approaches discard
the uncertainty. In this work, we propose to retain the probabilistic
models produced by OCR process in a relational database management
system. A key technical challenge is that the probabilistic data
produced by OCR software is very large (a single book blows up to 2GB
from 400kB as ASCII). As a result, a baseline solution that integrates
these models with an RDBMS is over 1000x slower versus standard text
processing for single table select-project queries. 
However, many applications may have quality-performance
needs that are in between these two extremes of ASCII and the
complete model output by the OCR software. Thus, we propose a novel
approximation scheme called \name that allows a user to trade recall
for query performance. Additionally, we provide a formal analysis of
our scheme's properties, and describe how we integrate our scheme with
standard-RDBMS text indexing.
\end{abstract}

\section{Introduction}

The mass digitization of books, printed documents, and printed forms
is changing the types of data that companies and academics manage. For
example, Google Books and their academic partner, the Hathi Trust,
have the goal of digitizing all of the world's books to allow scholars
to search human knowledge from the pre-Web era. The hope of this
effort is that digital access to this data will enable scholars to
rapidly mine these vast stores of text for new
discoveries.\footnote{Many repositories of {\it Digging into Data
    Challenge} (a large joint effort to bring together social
  scientists with data analysis) are OCR-based
  \url{http://www.diggingintodata.org}.}  The potential users of this
new content are not limited to academics. The market for {\em
  enterprise document capture} (scanning of forms) is already in the
multibillion dollar range~\cite{cmswire}. In many of the applications,
the translated data is related to enterprise business data, and so
after converting to plain text, the data are stored in an
RDBMS~\cite{exper:site}.

Translating an image of text (e.g., a jpeg) to ASCII is difficult for
machines to do automatically. To cope with the huge number of
variations in scanned documents, e.g., in spacing of the glyphs and
font faces, state-of-the-art approaches for optical character
recognition (OCR) use probabilistic techniques. For example, the
OCRopus tool from Google Books represents the output of the OCR
process as a stochastic automaton called a {\em finite-state
  transducer} (FST) that defines a probability distribution over all
possible strings that could be represented in the
image.\footnote{\url{http://code.google.com/p/ocropus/}.} An example
image and its resulting (simplified) transducer are shown in
Figure~\ref{fig:examplefst}. Each labeled path through the transducer
corresponds to a potential string (one multiplies the weights along
the path to get the probability of the string). Only to produce the
final plain text do current OCR approaches remove the
uncertainty. Traditionally, they choose to retain only the single most
likely string produced by the FST (called a {\em maximum a priori
  estimate} or MAP~\cite{murphy}).

\begin{comment}
\begin{figure*}
\centering
\vspace{-25pt}
\begin{tabular}{| >{\centering\arraybackslash}m{4.5cm} | >{\centering\arraybackslash}m{6cm}| >{\centering\arraybackslash}m{4.5cm} |}
\hline
& & \\
 \includegraphics[width=4cm]{images/newfordtextboxed2} & 
 \includegraphics[width=6cm]{images/newfordexamplefst} & 
\parbox{4.5cm}{ 
\begin{alltt}
SELECT DocID, Loss \\
FROM Claims\\
WHERE Year = 2010 AND\\
      DocData LIKE `\%Ford\%'
\end{alltt}
}
\\
(A) & (B) & (C)\\
\hline
\end{tabular}
\vspace{-5pt}
\caption{(A) An image of text. (B) A portion of a simple FST resulting
  from the OCR of the highlighted part of (A). The numbers on the arcs
  are conditional probabilities of transitioning from one state to
  another.  An emitted string corresponds to a path from states 0 to
  5. The string {\it `F0 rd'} (highlighted path) has the highest
  probability, $0.8*0.6*0.6*0.8*0.9 \approx 0.21$. (C) An SQL query to
  retrieve loss information that contains {\it `Ford'}. Using the MAP
  approach, no string is returned. Using \name, a match is obtained
  (probability $0.12$). }
\label{fig:examplefst}
\vspace{-20pt}
\end{figure*}
\end{comment}

\begin{figure*}[hbtp]
\centering
\includegraphics[width=6.8in]{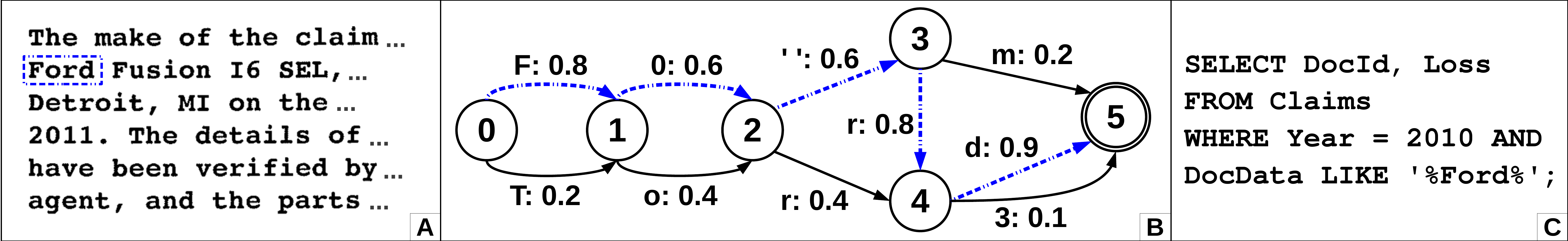}
\caption{(A) An image of text. (B) A portion of a simple FST resulting
  from the OCR of the highlighted part of (A). The numbers on the arcs
  are conditional probabilities of transitioning from one state to
  another.  An emitted string corresponds to a path from states 0 to
  5. The string {\it `F0 rd'} (highlighted path) has the highest
  probability, $0.8*0.6*0.6*0.8*0.9 \approx 0.21$. (C) An SQL query to
  retrieve loss information that contains {\it `Ford'}. Using the MAP
  approach, no claim is found. Using \name, a claim is found
  (with probability $0.12$). }
\label{fig:examplefst}
\vspace{-10pt}
\end{figure*}

As Google Books demonstrates, the MAP works well for browsing
applications. In such applications, one is sensitive to precision
(i.e., are the answers I see correct), but one is insensitive to
recall (i.e., what fraction of all of the answers in my corpus are
returned). But this is not true of all applications: an English
professor looking for the earliest dates that a word occurs in a
corpus is sensitive to recall~\cite{arun1}. As is an insurance company
that wants all insurance claims that were filled in 2010 that
mentioned a {\it `Ford'}. This latter query is expressed in SQL in
Figure~\ref{fig:examplefst}(C). 
In this work, we focus on such single table select-project queries, whose
outputs are standard probabilistic RDBMS tables.
Using the MAP approach may miss
valuable answers. In the example in Figure~\ref{fig:examplefst}, the
most likely string does not contain {\it `Ford'}, and so we
(erroneously) miss this claim. However, the string {\it `Ford'} does
appear (albeit with a lower probability). Empirically, we show that the
recall for simple queries on real-world OCR can be as low as $0.3$ --
and so we may throw away almost $70\%$ of our data if we follow the
MAP approach.

To remedy this recall problem, our baseline approach is to store and
handle the FSTs as binary large objects inside the RDBMS. As with a
probabilistic relational database, the user can then pose questions as
if the data are deterministic and it is the job of the system to
compute the confidence in its answer. By combining existing
open-source tools for transducer
composition \footnote{OpenFST. \url{http://www.openfst.org/}} with an
RDBMS, we can then answer queries like that in
Figure~\ref{fig:examplefst}(C). This approach achieves a high quality
(empirically, the recall we measured is very close to $1.0$, with up
to $0.9$ precision). Additionally, the enterprise users can ask their
existing queries directly on top of the RDBMS data (the query in
Figure~\ref{fig:examplefst}(C) remains unchanged). The downside is
that query processing is much slower (up to 1000x slower). While the
query processing time for transducers is linear in the data size, the
transducers themselves are huge, e.g., 
a single 200-page book blows up from 400 kB as text to over 2 GB
when represented by transducers after OCR.
This
motivates our central question: {\it ``Can we devise an approximation
  scheme that is somewhere in between these two extremes of recall and
  performance?''}

State-of-the-art OCR tools segment each of the images corresponding to
pages in a document into lines using special purpose line-breaking
tools. Breaking a single line further into individual words is more
difficult (spacing is very difficult to accurately detect). With this
in mind, a natural idea to improve the recall of the MAP approach is
to retain not only the highest probability string for each line, but
instead to retain the $k$ highest probability strings that appear in
each line (called $k$-MAP~\cite{approx-ai, map-relax}). Indeed, this
technique keeps more information around at a linear cost (in $k$) in
space and processing time. However, we show that even storing hundreds
of paths makes an insignificant jump in the recall of queries. 

To combat this problem, we propose a novel approximation scheme called
\name, which is our main technical contribution. The main idea is to
apply $k$-MAP not to the whole line, but to first break the line into
smaller chunks which are themselves transducers and apply $k$-MAP to
each transducer individually. This allows us to store exponentially
more alternatives than $k$-MAP (exponential in the number of chunks),
while using roughly a linear amount more space than the MAP
approach. If there is only a single chunk, then \name's output is
equivalent to $k$-MAP. If essentially every possible character is a
chunk, then we retain the full FST. Experimentally, we demonstrate
that the \name approach \textit{gracefully trades off between performance 
and recall}. For example, when looking for mentions of laws on a data set
that contains scanned acts of the US congress, the MAP approach
achieves a recall of $0.28$ executing in about $1$ second, the full
FST approach achieves perfect recall but takes over $2$ minutes. An
intermediate representation from \name takes around $10$ seconds and
achieves $0.76$ recall. Of course, there is a fundamental trade off
between precision and recall. On the same query as above, the MAP has
precision $1.0$, and the full FST has precision $0.25$, while \name
achieves $0.73$. In general, \name's precision falls in between the
MAP and the full FST.

To understand \name's approximation more deeply, we conduct a formal
analysis, which is our second technical contribution. When
constructing \name's approximation, we ensure two properties (1) each
chunk forms a transducer (as opposed to a more general structure), and
(2) that the model retains the {\em unique path property}, i.e., that
every string corresponds to a unique path. While both of these
properties are satisfied by the transducers produced by OCRopus,
neither property is necessary to have a well-defined approximation
scheme. Moreover, enforcing these two properties increases the
complexity of our algorithm and may preclude some compact
approximations. Thus, it is natural to wonder if we can relax these
two properties. While we cannot prove that these two conditions are
necessary, we show that without these two properties, basic operations
become intractable. Without the unique path property, prior work has
shown that determining (even approximating) the $k$-MAP is intractable
for a fixed $k$~\cite{re-transducers}. Even with the unique path
property and a fixed set of chunks, we show that essentially the
simplest violation of property $(1)$ makes it intractable to construct
an approximation even for $k = 2$ (Theorem~\ref{thm:hard}). On the
positive side, for any fixed partition, \name retains a set of strings
that achieves the highest total probability among approximations that
satisfy the above restrictions.

Finally, we describe how to use standard text-indexing techniques to
improve query performance. Directly applying an inverted index to
transducer data is essentially doomed to failure: the sheer number of
terms one would have to index grows exponentially with the length of
the document, e.g., an FST for a single line may represent over
$10^{100}$ terms. To combat this, we allow the user to specify a
dictionary of terms. We then construct an index of those terms
specified in the dictionary. This allows us to process keyword and
some regular expressions using standard
techniques~\cite{classical1,classical2}.

\paragraph*{Outline} 
In Section~\ref{sec:prelim}, we illustrate our current prototype system
to manage OCR data using an RDBMS with an example, and we present a
brief background on the use of transducers in OCR. 
In Section~\ref{sec:scans}, we briefly describe the baseline solutions, and then 
discuss the main novel technical contributions of this work, viz., the \name 
approximation scheme and our formal analysis of its properties.
In Section~\ref{sec:index}, we describe our approach for indexing OCR transducer 
data, which is another technical contribution of this work. In Section 
\ref{sec:experiments}, we empirically validate that our approach is able
to trade off recall for query-runtime performance on several
real-world OCR data sets. We validate that our approximation methods
can be efficiently implemented, and that our indexing technique
provides the expected speedups. In
Section~\ref{sec:related_work}, we discuss related work.
\eat{In Section~\ref{sec:scans}, we describe the \name approximation and our
formal analysis of its properties. In Section~\ref{sec:index}, we
describe our approach for indexing FSTs.} 

\section{Preliminaries}
\label{sec:prelim}

The key functionality that \name provides is to enable users to query
OCR data inside an RDBMS as if it were regular text. Specifically, we
want to enable the LIKE predicate of SQL on OCR data. We describe
\name through an example, followed by a more detailed explanation of
its semantics and the formal background. 

\subsection{Using Staccato with OCR}

Consider an insurance company that stores loss data with scanned
report forms in a table with the following schema:
\[\textbf{Claims} (DocID,Year, Loss, DocData) \]

\noindent
A document tuple contains an id, the year the form was filed (Year),
the amount of the loss (Loss) and the contents of the report
(DocData). A simple query that an insurance company may want to ask
over the table - {\it ``Get loss amounts of all claims in 2010 where
  the report mentions `Ford' ''}. Were DocData ASCII text, this could
be expressed as an SQL query as follows:

\begin{alltt}
SELECT DocID, Loss FROM Claims 
WHERE Year = 2010 AND DocData LIKE `%Ford\%'; 
\end{alltt}
        
If DocData is standard text, the semantics of this query is
straightforward: we examine each document filed in 2010, and
check if it contains the string {\it `Ford'}. The challenge is that
instead of a single document, in OCR applications DocData represents
many different documents (each document is weighted by
probability). In \name, we can express this as an SQL query that uses a
simple pattern in the \texttt{LIKE} predicate (also in
Figure~\ref{fig:examplefst}(C)). The twist is that the underlying
processing must take into account the probabilities from the OCR
model.

Formally, \name allows a larger class of queries in the LIKE predicate
that can be expressed as deterministic finite automata (DFAs). \name
translates the syntax above in to a DFA using standard
techniques~\cite{hopcroft}. As with probabilistic
databases~\cite{DBLP:conf/vldb/DalviS04,trio, amol-query,
  olteanu-maybms}
  %,yanlei-model}
  , \name computes the probability that
the document matches the regular expression. \name does this using
algorithms from prior work~\cite{re-transducers, re-query}. The result
is a probabilistic relation; after this, we can apply probabilistic
relational database processing
techniques~\cite{DBLP:conf/vldb/DalviS04,DBLP:conf/icde/OlteanuHK09,
  DBLP:conf/icde/SarmaTW08}. In this work, we consider only single
table select-project queries (joins are handled using the above
mentioned techniques).

A critical challenge that \name must address is given a DFA find those
documents that are relevant to the query expressed by the DFA. For a
fixed query, the existing algorithms are roughly linear in the size of
data that they must process. To improve the runtime of these
algorithms, one strategy (that we take) is to reduce the size of the
data that must be processed using approximations. The primary
contribution of \name is the set of mechanisms that we describe in
Section~\ref{sec:scans} to achieve the trade off of quality and
performance by approximating the data. We formally study the
properties of our algorithms and describe simple mechanisms to allow
the user to set these parameters in Sec.~\ref{sec:analysis}.

One way to evaluate the query above in the deterministic setting is to
scan the string in each report and check for a match. A better
strategy may be to use an inverted index to fetch only those documents
that contain {\it `Ford'}. In general, this strategy is possible for
{\em anchored} regular expressions~\cite{regexp}, which are regular expressions 
that begin or end with words in the language, e.g. `no.(2$\vert$3)' is 
anchored while `(no$\vert$num).(2$\vert$8)' is not.
\name supports a
similar optimization using standard text-indexing techniques. There
is, however, one twist: At one extreme, any term may have some small
probability of occurring at every location of the document -- which
renders the index ineffective. Nevertheless, we show that 
\name is able to provide efficient indexing for anchored regular 
expressions using a dictionary-based approach.

\subsection{Background: Stochastic Finite Automata}

We formally describe \name's data model that is based on Stochastic
Finite Automata (SFA). This model is essentially
identical to the model output by Google's OCRopus~\cite{ocropus,
  mohri-transducers}.\footnote{Our prototype uses the same weighted
  finite state transducer (FST) model that is used by OpenFST and
  OCRopus. We simplify FST to SFAs here only slightly for
  presentation. See the full version for more
  details~\cite{full-paper} } An SFA is a finite state machine that
emits strings (e.g., the ASCII conversion of an OCR image). The model
is stochastic, which captures the uncertainty in translating the glyphs
and spaces to ASCII characters.

At a high level, an SFA over an alphabet $\Sigma$ represents a discrete
probability distribution $P$ over strings in $\Sigma^*$, i.e., 
\[ P : \Sigma^{*} \to [0,1] \text{ such that } \sum_{x \in \Sigma^{*}} P(x) = 1 \]
The SFA represents the (finitely many) strings with non-zero
probability using an automaton-like structure that we first describe
using an example:\\

\begin{example}
  Figure~\ref{fig:examplefst} shows an image of text and a
  simplified SFA created by OCRopus from that data. The SFA is a
  directed acyclic labeled graph. The graphical structure (i.e., the
  branching) in the SFA is used by the OCR tool to capture
  correlations between the emitted letters. Each source-to-sink path
  (i.e., a path from node 0 to node 5) corresponds to a string with
  non-zero probability. For example, the string {\it `Ford'} is one
  possible path that uses the following sequence of nodes $0 \to 1 \to
  2 \to 4 \to 5$. The probability of this string can be found by
  multiplying the edge weights corresponding to the path:
  $0.8*0.4*0.4*0.9 \approx 0.12$.
\end{example}
\newline

Formally, we fix an alphabet $\Sigma$ (in \name, this is the set of
ASCII characters). An SFA $S$ over $\Sigma$ is a tuple $S=(V,E,s,
f,\delta)$ where $V$ is a set of nodes, $E \subseteq V \times V$ is a
set of edges such that $(V,E)$ is a directed acyclic graph, and $s$
(resp. $f$) is a distinguished start (resp. final) node. The function
$\delta$ is a stochastic transition function, i.e.,
%\[ \delta : E \times \Sigma \to [0,1] \text{ such that } \sum_{\stackrel{(x,y) \in E}{\sigma \in \Sigma}} \delta((x,y),\sigma) = 1 \quad  \text { for each } x \in V \]
\[ \delta : E \times \Sigma \to [0,1] \text{ s.t.} \sum_{\stackrel{y: (x,y) \in E}{\sigma \in \Sigma}} \delta((x,y),\sigma) = 1 \quad  \forall x \in V \]
In essence, $\delta( e, \sigma)$, where $e = (x,y)$, is the conditional
probability of transitioning from $x \to y$ and emitting $\sigma$.

An SFA defines a probability distribution via its labeled paths.  A
labeled path from $s$ to $f$ is denoted by $p = (e_1,\sigma_1), \dots,
(e_{N},\sigma_{N})$, where $e_i \in E$ and $\sigma_i \in \Sigma$,
corresponding to the string $\sigma_1...\sigma_n$, with its
probability:
\footnote{Many (including OpenFST) tools use a formalization with
  log-odds instead of probabilities. It has some intuitive property for
  graph concepts, e.g., the shortest path corresponds to the most
  likely string.}
\[ \Pr_S[p] = \prod_{i=1}^{|p|} \delta( e_i ,\sigma_i)  \]

SFAs in OCR satisfy an important property that we call the
\textit{unique paths property} that says that any string produced by
the SFA with non-zero probability is generated by a unique labeled
path through the SFA. We denote by $\mathrm{UP}$ the function that
takes a string to its unique labeled path. This property guarantees
tractability of many important computations over SFAs including
finding the highest probability string produced by the
SFA~\cite{re-transducers}.

Unlike the example given here, the SFAs produced by Google's OCRopus are much
larger: they contain a weighted arc for every ASCII character. And so,
the SFA for a single line can require as much as 600 kB to store.

Queries in \name are formalized in the standard way for probabilistic
databases. In this paper, we consider LIKE predicates that contain
Boolean queries expressed as DFAs (\name handles non-Boolean queries
using algorithms in Kimmelfeld and
R\'e~\cite{re-transducers}). Fix an alphabet $\Sigma$ (the ASCII
characters). Let $q : \Sigma^{*} \to \set{0,1}$ be expressed as DFA
and $x$ be any string. 
We have $q(x)=1$ when $x$ satisfies the query, i.e., it's accepted by the DFA.
We compute the probability that $q$ is true;
this quantity is denoted $\Pr[q]$ and is defined by $\Pr[q] = \sum_{x
  \in \Sigma^{*}} q(x) \Pr(x)$ (i.e., simply sum over all possible
strings where $q$ is true). There is a straightforward algorithm based
on matrix multiplication to process these queries that is linear in the
size of the data and cubic in the number of states of the DFA~\cite{re-query}.

\section{Managing SFAs in an RDBMS}
\label{sec:scans}

We start by outlining two baseline approaches that represent the two
extremes of query performance and recall. Then, we describe the novel
approximation scheme of \name, which enables us to trade performance
for recall.

\paragraph*{Baseline Approaches} 
We study two baseline approaches: $k$-MAP and the FullSFA approach.
Fix some $k \geq 1$. In the $k$-MAP approach we store the $k$ highest
probability strings (simply, top $k$ strings) generated by each SFA in
our databases. We store one tuple per string along with the associated
probability. Query processing is straightforward: we process each
string using standard text-processing techniques, and then sum the
probability of each string (since each string is a disjoint
probabilistic event). In the FullSFA approach, we store the entire SFA
as a BLOB inside the RDBMS. To answer a query, we retrieve the BLOB,
deserialize it, and then use an open source C++ automata composition
library to answer the query \cite{mohri-openfst, mohri-composition}
and compute all probabilities. Table \ref{tab:complexities} summarizes
the time and space costs for a simple chain SFA (no branching). This
table gives an engineer's intuition about the time and space
complexity of the baseline approaches. The factor 16 accounts for the
metadata -- tuple ID, location in SFA, and probability value (the
schema is described in the full version \cite{full-paper}).  We also
include our proposed approach, \name that depends on a parameter $m$
(the number of chunks) that we describe below. From the table, we can
read that query processing time for 
\name
is essentially linear in $m$. Let
$l$ be the length of the document, since $m \in [1,l]$ query
processing time in \name interpolates linearly from the $k$-MAP approach 
to the FullSFA approach.

\begin{table}[hbtp]
\centering
%\begin{tabular}{|p{0.9cm}|p{0.7cm}|p{1.1cm}|p{1.7cm}||p{2.0cm}||}
\begin{tabular}{|l|l|l||l|}
\hline
& $k$-MAP & FullSFA & \name \\
\hline
\hline
\textsc{Query}&  $l q k$ & $l q |\Sigma| + q^3(l -1)$ & $l q k + q^3(m - 1)$\\
\hline
\textsc{Space}& $l k + 16 k$ & $l|\Sigma|+ 16l |\Sigma|$ & $l k + 16mk$\\
\hline
\end{tabular}\\
\begin{tabular}{ccp{7cm}}
%\hline
%\multicolumn{2}{|c|}{\textbf{Notation}} \\
%\hline
&\\
$l$ & : & length of the SFA's strings\\
%\hline
$q$ & : & \# states in the query DFA\\
%\hline
$k$ & : & \# paths parameter in $k$-MAP, \name \\
%\hline
$m$ & : & \# chunks in \name ($1 \leq m \leq l$)\\
%\hline
\end{tabular}
\caption{Space costs and query processing times for a simple
 chain SFA. The space indicates the number of bytes of storage required.}
\label{tab:complexities}
\end{table}

\subsection{Approximating an SFA with Chunks}
\label{sec:approximation}
\eat{The SFA for a single 200-page book may take over 2 GB to store.}
As mentioned before, the SFAs in OCR are much larger than our example, 
e.g. one OCR line from a scanned book yielded an SFA of size 600 kB.
In turn, the 200-page book blows up to over 2 GB when represented by SFAs.
Thus,
to answer a query that spans many books in the FullSFA approach, we
must read a huge amount of data. This can be a major bottleneck in
query processing. To combat this we propose to approximate 
an SFA
with a collection of smaller-sized SFAs (that we call {\em
  chunks}). Our goal is to create an approximation that allows us to
gracefully tradeoff from the fast-but-low-recall 
MAP approach to
the slow-but-high-recall FullSFA approach.
 
\eat{
OCR tools, like OCRopus, do partition the input into logical units
(typically lines). Thus, a natural first approximation to an SFA is to 
simply store the top-$k$ paths in each of the resulting per-line SFAs.
}
Recall that the $k$-MAP approach is a natural first approximation,
wherein we simply store the top-$k$ paths in each of the per-line SFAs.
This approach can increase the recall at a linear cost in $k$. However, 
as we demonstrate experimentally, simply increasing $k$ is insufficient
to tradeoff between the two extremes. That is, even for huge values
of $k$ we do not achieve full recall.

\begin{figure}
%\textbf{TOP $K$ just wiggles on one character, we break it into
 % blocks. Exponentially better.}
\centering
\includegraphics[width=3.3in]{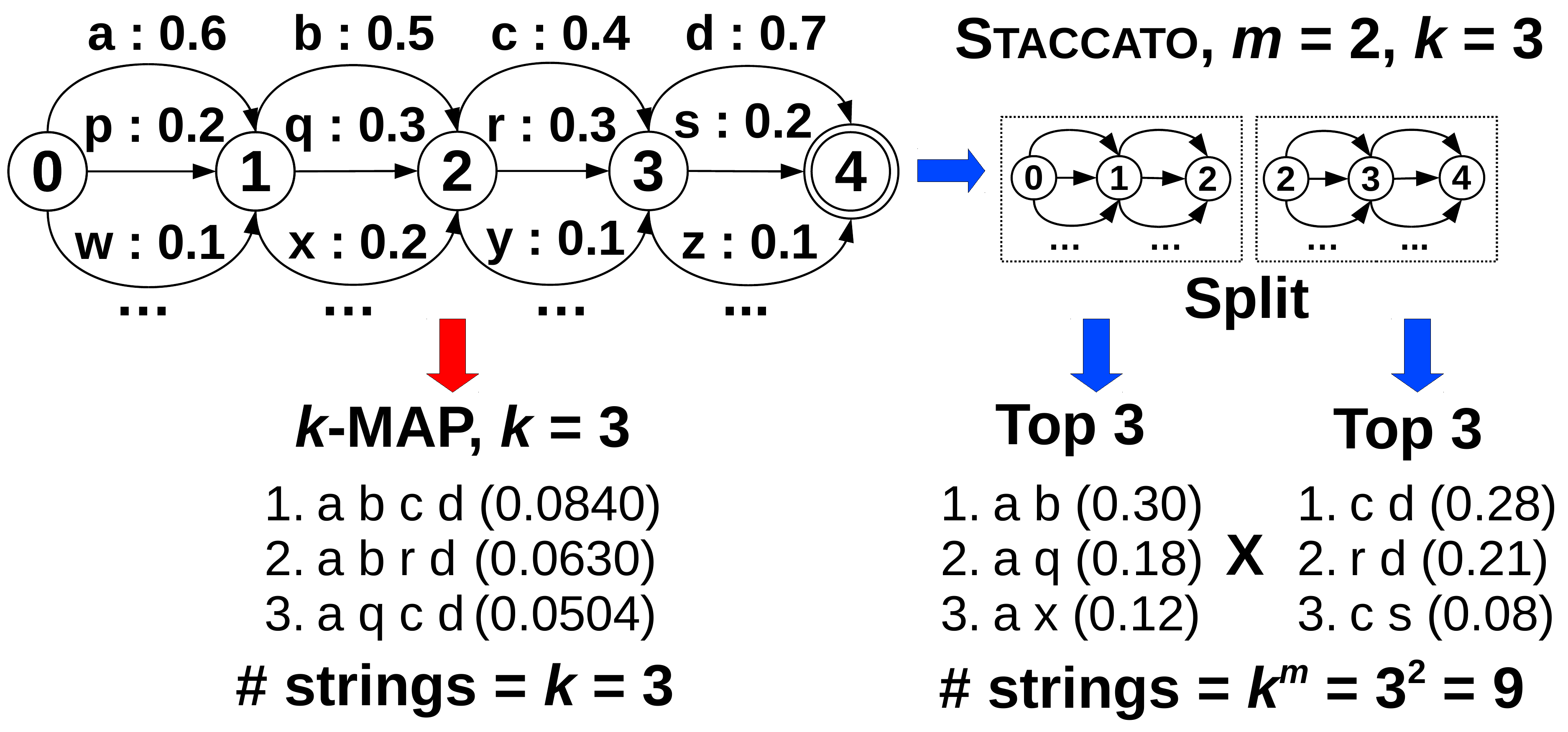}
\caption{A depiction of conventional Top-$k$ versus \name's approximation.}
\label{fig:kmapvsstaccato}
\end{figure}

\begin{figure*}[hbtp]
\centering
\submissionversion{
\begin{minipage}[t]{2.6in}
}
\fullversion{
\begin{minipage}[t]{4.2in}
}
\centering
\begin{algorithm}[H]
\textbf{Inputs}: SFA $S$ with partial order $\leq$ on its nodes, $X \subseteq V$\\
\While {$X$ does not form a valid SFA} { 
	\uIf{\textit{No unique start node in $X$}} 
  		{\textit{Compute the least common ancestor of $X$, say, $l$\\
  			$X \leftarrow X \cup \{y \in V ~|~ l \leq y ~and~ $ 
  			$\forall x \in X, y \leq x\}$}}
	\uIf{\textit{No unique end node in $X$}} 
		{\textit{Compute greatest common descendant of $X$, say, $g$\\
			$X \leftarrow X \cup \{y \in V ~|~ y \leq g ~and~ $
			$\forall x \in X, x \leq y\}$}}
	\textit{$\forall e \in E$ s.t. exactly one end-point is in $X-\{l,g\}$, add other end-point to $X$}
}
\caption{FindMinSFA}
\label{alg:findminsfa}
\end{algorithm} 
%end eat algo
\end{minipage}
\begin{minipage}{4.29in}
\includegraphics[width=4.29in]{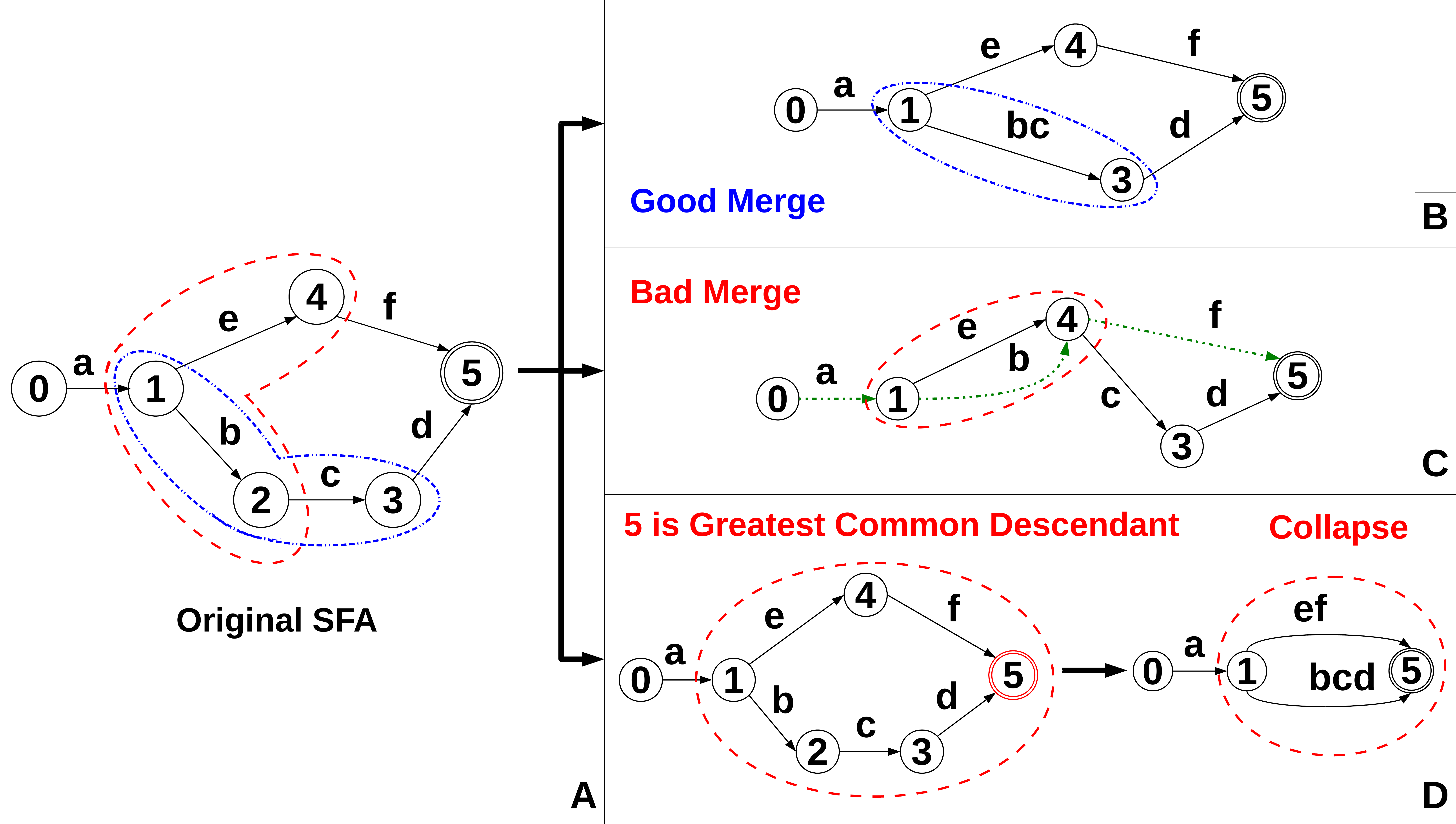}
\end{minipage}
\caption{
Algorithm \ref{alg:findminsfa}: FindMinSFA.
Illustrating merge and FindMinSFA: (A) Original: The SFA emits
two strings: $aef$ and $abcd$.  Two merges considered: \{(1,2),(2,3)\} 
(successive edges), and \{(1,2),(1,4)\} (sibling edges).
(B) Good merge: First set gives new edge (1,3), emitting $bc$. The SFA still emits 
only $aef$ and $abcd$. 
(C) Bad merge: Second set gives new edge (1,4), emitting $e$ and $b$. 
But, the SFA now wrongly emits
new strings, e.g., $abf$ (dashed lines).
(D) Using Algorithm \ref{alg:findminsfa} on the second set, the
greatest common descendant is obtained (node 5), and the resulting set is collapsed to edge (1,5).
The SFA now emits only $aef$ and $abcd$.
}
\label{fig:fixup}
%\vspace{-20pt}
\end{figure*}

Our idea to combat the slow increase of recall starts with the
following intuition: the more strings from the SFA we store, the
higher our recall will be. We observe that if we store the top $k$ in
each of $m$ smaller SFAs (that we refer to as `chunks'), 
we effectively store $k^m$ distinct strings. Thus,
increasing the value of $k$ increases the number of strings
polynomially. In contrast, increasing $m$, the number of smaller SFAs,
increases the number of paths {\em exponentially}, as illustrated in
Figure \ref{fig:kmapvsstaccato}. This observation motivates the idea
that to improve quality, we should divide the SFA further.
%beyond what is
%necessary to fit it in memory or logically required by the user.  
As we demonstrate experimentally, \name achieves the most conceptually
important feature of our approximation: it allows us to smoothly tradeoff recall for performance. In other words, increasing $m$ (and $k$)
increases the recall at the expense of performance.

\paragraph*{SFA Approximation}

Given an SFA $S$, our goal is to find a new SFA $S'$ that satisfies
two competing properties: (1) $S'$ should be smaller than $S$, and (2)
the set of strings represented by $S'$ should contain as many of the
high probability strings from $S$ as possible without containing any
strings not in $S$.\footnote{This is a type of {\em sufficient lineage
    approximation}~\cite{re-lineage}.} Our technique to approximate
the SFA $S$ is to merge a set of transitions in $S$ (a `chunk') to
produce a new SFA $S'$; then we retain only the top $k$ transitions on
each edge in $S'$.

To describe our algorithm, we need some notation. We generalize the
definition of SFAs (Section \ref{sec:prelim}) to allow transitions
that produce strings (as opposed to single characters).  Formally, the
transition function $\delta$ has the type $\delta : E \times
\Sigma^{+} \rightarrow [0,1]$.  Any SFA meets this generalized SFA
definition, and so we assume this generalized definition of SFAs for
the rest of the section.

Before describing the merging operation formally, we illustrate the
challenge in the merging process in Figure \ref{fig:fixup}. Figure
\ref{fig:fixup}(A) shows an SFA (without probabilities for
readability). We consider two merging operations.
First, we have chosen to merge the edges $(1,2)$ and 
$(2,3)$ and replaced it with a single edge $(1,3)$. To retain the 
same strings that are present in the SFA in (A), 
the transition function must emit the string `$bc$' on
the new edge $(1,3)$ as illustrated in Figure~\ref{fig:fixup}(B). In
contrast, if we choose to merge the edges $(1,2)$ and $(1,4)$, there
is an issue: {\it no matter what we put on the transition from $(1,4)$
  we will introduce strings that are not present in the original
  SFA} (Figure \ref{fig:fixup}(C)). The problem is that the set of nodes $\set{1,2,4}$ do not
form an SFA by themselves (there is no unique final node). One could
imagine generalizing the definition of SFA to allow richer structures
that could capture the correlations between strings, but as we explain
in Section~\ref{sec:analysis}, this approach creates serious technical
challenges. Instead, we propose to fix this issue by searching for a
minimal SFA $S'$ that contains this set of nodes (the operation called
\textsc{FindMinSFA}).  
Then, we replace the nodes in the set with a
single edge, retaining only the top $k$ highest probability strings
from $S'$. We refer to this operation of replacing $S$' with an edge
as \textsc{Collapse}.  In our example, the result of these operations
is illustrated in Figure \ref{fig:fixup}(D).

We describe our algorithm's subroutine \textsc{FindMinSFA} and then the entire
heuristic.

\paragraph*{FindMinSFA} 
Given an SFA $S$ and a set of nodes $X \subseteq V$, our goal is to
find a SFA $S'$ whose node set $Y$ is such that that $X \subseteq
Y$. We want the set $Y$ to be minimal in the sense that removing any
node $y \in Y$ causes $S'$ to violate the SFA property, that is
removing $y$ causes $S'$ to no longer have a unique start (resp. end)
state. Presented in Algorithm \ref{alg:findminsfa}, 
our algorithm is based on the observation that the unique start
node $s$ of $S'$ must come before all nodes in $X$ in the topological
order of the graph (a partial order). 
Similarly, the end node $f$ of the SFA $S'$ must
come after all nodes in $X$ in topological order. To satisfy these
properties, we repeatedly enlarge $Y$ by computing the start
(resp. final node) using the least common ancestor (resp. greatest
common descendant) in the DAG. Additionally, we require that any edge
in $S$ that is incident to a node in $Y$ can be incident to only
either $s$ or $f$. (Any node incident to both will be internal to
$S'$) If there are no such edges, we are done. Otherwise, for each
such edge $e$, we include its endpoints in $Y$ and repeat this
algorithm with $X$ enlarged to $Y$. Once we find a suitable set $Y$,
we replace the set of nodes in the SFA $S$ with a single edge from $s$
(the start node of $S'$) to $f$ (the final node of $S'$).
Figure \ref{fig:fixup}(D) illustrates
a case when there is no unique end node, and the greatest common
descendant has to be computed. More illustrations, covering
the other cases, are presented in the full version \cite{full-paper}.
%Figure \ref{fig:fixup}(D) is the case 
%when an external edge is incident upon an internal node, and thus 
%has to be assimilated.

\eat{\begin{algorithm}[hbtp]
\textbf{Inputs}: SFA $S$ with partial order $\leq$ on nodes, $X \subseteq V$\\
\While {$X$ does not form a valid SFA} { 
	\uIf{\textit{No unique start node in $X$}} 
  		{\textit{Compute least common ancestor ($l$) of $X$\\
  			$X \leftarrow X \cup \{y \in V ~|~ l \leq y ~and~ \forall x \in X, y \leq x\}$}}
	\uIf{\textit{No unique end node in $X$}} 
		{\textit{Compute greatest common descendant ($g$) of $X$\\
			$X \leftarrow X \cup \{y \in V ~|~ y \leq g ~and~ \forall x \in X, x \leq y\}$}}
	\textit{$\forall e \in E$ with exactly one end-point of $e$ in $X-\{l,g\}$, add the other end-point to $X$}
}
\caption{\textsc{FindMinSFA}}
\label{alg:findminsfa}
\end{algorithm}
}

\paragraph*{Algorithm Description}
The inputs to our algorithm are the parameters $k$ (the number of strings
retained per edge) and $m$ (the maximum number of edges that we are
allowed to retain in the resulting graph). We describe how a user
chooses these parameters in Section~\ref{sec:analysis}. For now, we
focus on the algorithm. At each step, our approximation creates a
restricted type of SFA where each edge emits at most $k$ strings, 
i.e., $\forall e \in E$, $|\setof{ \sigma \in \Sigma^{*} }{ \delta(e,
  \sigma) > 0 }| \leq k$. When given an SFA not satisfying this property,
our algorithm chooses to retain those strings $\sigma \in \Sigma^{*}$
with the highest values of $\delta$ (ties broken arbitrarily). This set
can be computed efficiently using the standard Viterbi algorithm 
\cite{viterbi}, which is a  dynamic programming 
algorithm for finding the most likely outputs in
probabilistic sequence models, like HMMs. By memoizing the best partial
results till a particular state, it can compute the globally optimal 
results in polynomial time. To compute the top-$k$ results more 
efficiently, we use an incremental variant by Yen et al \cite{yen}.

\begin{algorithm}[hbtp]
\textit{Choose $\set{x,y,z}$ s.t. $(x,y),(y,z)\in E$ and maximizing
  the probability mass of the retained strings.}\\ $S \leftarrow
\textsc{Collapse}(\textsc{FindMinSFA}(S,\set{x,y,z})
)$\\ \textit{Repeat above steps till $|E| \le m$}\\
\caption{Greedy heuristic over SFA $S = (V,E)$}
\label{alg:greedymerge}
\end{algorithm}

Algorithm \ref{alg:greedymerge} summarizes our heuristic: for each
triple of nodes $\set{x,y,z}$ such that $(x,y),(y,z) \in E $,
we find a minimal containing SFA $S_{ij}$ by calling
$\textsc{FindMinSFA}( \set{x,y,z} )$. We then replace the set of nodes
in $S_{ij}$ by a single edge $f$ (\textsc{Collapse} above). This edge
$f$ keeps only the top-$k$ strings produced by $S_{ij}$. Thus, the
triple of nodes $\set{x,y,z}$ generates a candidate SFA. We choose the
candidate such that the probability mass of all generated strings is
as high as possible (note that since we have thrown away some strings,
the total probability mass may be less than $1$). Given an SFA we can
compute this using the standard sum-product algorithm (a faster
incremental variant is actually used in \name). We then continue to
recurse until we have reached our goal of finding an SFA that contains
fewer than $m$ edges. A simple optimization (employed by \name) is to
cache those candidates we have considered in previous iterations.

While our algorithm is not necessarily optimal, it serves as a proof
of concept that our conceptual goal can be achieved. That is, \name
offers a knob to tradeoff recall for performance. We describe the
experimental setup in more detail in Section~\ref{sec:experiments},
but we illustrate our point with a simple experimental result. Figure
\ref{fig:tradeoff} plots the recall and runtimes of the two
baselines and \name. Here, we have set $k = 100$ and $m = 10$.  On these two
queries, \name falls in the middle on both recall and performance. \\

\begin{figure}[hbtp]
\centering
\includegraphics[width=2.7in]{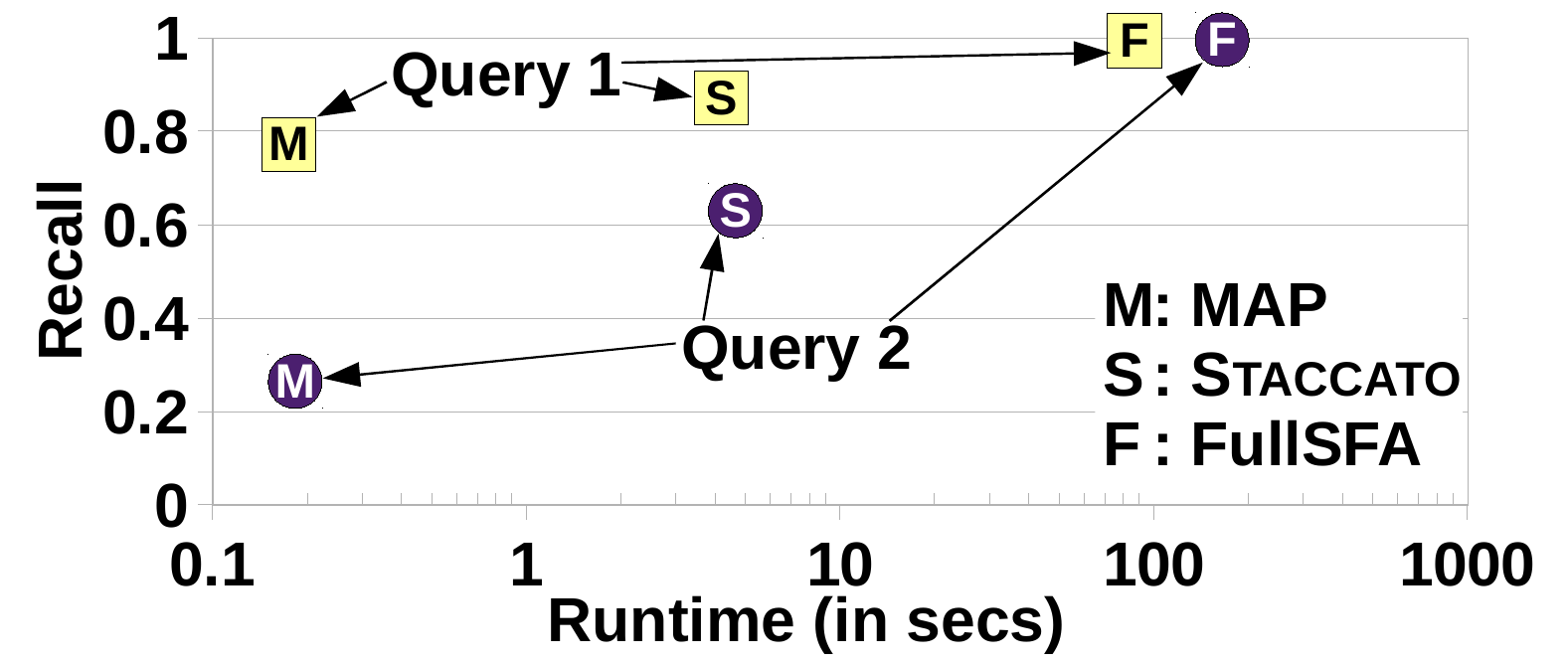}
\caption{Recall - Runtime tradeoff for a keyword query (Query 1) and a regular expression 
query (Query 2).
The parameters are: number of chunks ($m$) = 10, number of paths per chunk ($k$) = 100, 
and number of answers queried for ($NumAns$) = 100.
}
\label{fig:tradeoff}
\end{figure}

\eat{
\begin{figure}[hbtp]
\centering
\includegraphics[width=2.5in]{plots/AllResults/RecVsRuntime/recvsruntimes}
\caption{Recall - Runtime tradeoff for two queries across the three approaches - 
MAP (M), Staccato (S) and FullSFA (F). Query 1 (identified by squares) is 
a keyword query 
%$President$ 
 and 
query 2 (circles) is a regular expression query 
%$U.S.C.~2 \backslash d \backslash d 
%\backslash d$, where  $\backslash d$ is $(0|1...|9)$
.
The Staccato parameters used here 
are $m = 10,~k=100$, and $NumAns = 100$. 
%The ground truth numbers for these queries are 14 and 25, respectively.
}
\label{fig:tradeoff}
\end{figure}
}

\subsection{Extensions and Analysis}
\label{sec:analysis}

To understand the formal underpinning of our approach, we perform a
theoretical analysis.  Informally, the first question is: {\it ``in
  what sense is choosing the $k$-MAP the best approximation for each
  chunk in our algorithm?''}  The second question we ask is to justify
our restriction to SFAs as opposed to richer graphical structures in
our approximation. We show that $k$-MAP in each chunk is no longer the
best approximation and that there is likely no simple algorithm (as an
underlying problem is \NP-complete.) 

We formally define the goal of our algorithms. Recall that an SFA $S$
on $\Sigma$ represents a probability distribution $\Pr_S : \Sigma^{*}
\to [0,1]$. Given a set $X \subseteq \Sigma^{*}$, define $\Pr_{S}[X] =
\sum_{x \in X} \Pr_{s}[x]$. All the approximations that we consider
emit a subset of strings from the original model. Given an
approximation scheme $\alpha$, we denote by $\mathrm{Emit}(\alpha)$
the set of strings that are emitted (retained) by that scheme. All
other things being equal, we prefer a scheme $\alpha$ to $\alpha'$ whenever
\[ \Pr_{S}[\mathrm{Emit}(\alpha)] \geq \Pr_{S}[\mathrm{Emit}(\alpha')] \]
That is, $\alpha$ retains more probability mass than
$\alpha'$. 
The formal basis for this choice is a 
standard statistical measure called the 
{\em Kullback-Leibler Divergence}~\cite{bishop}, between the original and 
the approximate probability distributions. In the full 
version \cite{full-paper}, we show that this divergence is lower (which means 
the approximate distribution is more similar to the original distribution) 
if the approximation satisfies the above 
inequality. In other words, a better approximation retains more of the 
high-probability strings.

\eat{\footnote{In the full version, we further justify this
  choice using {\em Kullback-Leibler Divergence} ($\mathsf{KL}$
  divergence) also called {\em relative entropy}~\cite{bishop} from
  information theory.}
}

We now describe our two main theoretical results. First for SFAs,
\name's approach to choosing the $k$ highest probability strings in
each chunk is optimal. For richer structures than SFAs, finding the
optimal approximation is intractable (even if we are given the chunk
structure, described below). Showing the first statement is
straightforward, while the result about richer structures is more
challenging.

\paragraph*{Optimality of $k$-MAP for SFAs}
Given a generalized SFA $S=(V,\delta)$. Fix $k \geq 1$. Let
$\mathbf{S}_{[k]}$ be the set of all SFAs $(V,\delta')$ that arise
from picking $k$ strings on each edge of $S$ to store. That is, for
any pair of nodes $x,y \in V$ the set of strings with non-zero
probability has size smaller than $k$:
\[ \size{\setof{\sigma \in \Sigma^{*} }{\delta'((x,y),\sigma) > 0}}
\leq k \] 

\noindent
Let $S_k$ denote an SFA that for each pair $(x,y) \in V$ chooses the
highest probability strings in the model (breaking ties
arbitrarily). Then,

\begin{proposition}
For any $S' \in \mathbf{S}_{[k]}$, we have:
\[  
\Pr_{S}[ \mathrm{Emit}(S_k) ] \geq \Pr_{S}[ \mathrm{Emit}(S')  ]
 \]
\end{proposition}
Since $S_k$ is selected by \name, we view this as formal justification
for \name's choice. 

\paragraph*{Richer Structural Approximation} 
We now ask a follow-up question: {\it ``If we allow more general
  partitions (rather than collapsing edges), is $k$-MAP still
  optimal?''}  To make this precise, we consider a partition of the
underlying edges of the SFA into connected components (call that
partition $\Phi$). Keeping with our early terminology, an element of
the partition is called a {\em chunk}. In each chunk, we select at most $k$ strings
(corresponding to labeled paths through the chunk). Let $\alpha :
\Phi \times \Sigma^{*} \to \set{0,1}$ be an indicator function such
that $\alpha(\phi, \sigma) = 1$ only if in chunk $\phi$ we choose
string $\sigma$. For any $k \geq 1$, let $A_k$ denote the set of all
such $\alpha$s that picks at most $k$ strings from each chunk, i.e.,
for any $\phi \in \Phi$ we have $\size{\setof{\sigma \in
    \Sigma^{*}}{\alpha(\phi,\sigma) > 0}} \leq k$. Let
$\mathrm{Emit}(\alpha)$ be the set of strings emitted by this
representation with non-zero probability (all strings that can be
created from concatenating paths in the model).

Following the intuition from the SFA case described above, the best
$\alpha$ would select the $k$-highest probability strings in each
chunk. However, this is not the case. Moreover, we exhibit chunk
structures, where finding the optimal choice of $\alpha$ is \NP-hard
in the size of the structure. This makes it unlikely that there is any
simple description of the optimal approximation.

\begin{theorem}
  Fix $k \geq 2$. The following problem is $\NP$-complete. Given as
  input $(S,\Phi,\lambda)$ where $S$ is an SFA, $\Phi$ partitions the
  underlying graph of $S$, and $\lambda \geq 0$, determine if there
  exists an $\alpha \in A_k$ satisfying $\Pr[ \mathrm{Emit}(\alpha) ]
  \geq \lambda$.
\label{thm:hard}
\label{THM:HARD}
\end{theorem}

The above problem remains $\NP$-complete if $S$ is restricted to
satisfy the unique path property and restricted to a binary
alphabet. A direct consequence of this theorem is that finding the
maximizer is at least $\NP$-hard. 
We provide the proof of this theorem in the full version~\cite{full-paper}.
The proof includes a detailed 
outline of a reduction from a matrix multiplication-related problem that 
is known be to hard. The reduction is by a gadget construction that 
encodes matrix multiplication as SFAs.
Each chunk has at most $2$ nodes in
either border (as opposed to an SFA which has a single start and final
node). This is about the weakest violation of the SFA property that we
can imagine, and suggests to us that the SFA property is critical for
tractable approximations.

\paragraph*{Automated Construction of \name}

Part of our goal is to allow knobs to trade recall for performance on
a per application basis, but setting the correct values for $m$ and
$k$ may be unintuitive for users. 
To reduce the burden on the user, we devise a simple parameter 
tuning heuristic that maximizes query performance, while achieving acceptable recall. 
To measure recall, the user provides a set of labeled examples and representative queries.
%We use the size of the data as a surrogate for query performance.
The user specifies a quality constraint (average recall for the set of queries) 
and a size constraint (storage space as percentage of the original dataset size). 
%,also used as a surrogate for query performance). 
The goal is to find a pair of parameters ($m$,$k$) that satisfies both these constraints. 
We note that the size of the data is a function of ($m$,$k$) (see Table \ref{tab:complexities}),
which along with the size constraint helps us express $k$ in terms of $m$ (or vice versa).
We empirically observed that for a fixed size, a smaller $m$ usually yields faster query 
performance than a smaller $k$, which suggests that we need to minimize the value of $m$ to
maximize query performance.
%In contrast, we must estimate the recall from the data.
Our method works as follows: we pick a given value of $m$, then calculate the corresponding
$k$ that lies on the size constraint boundary. 
Given the resulting ($m$,$k$) pair, 
we compute the \name approximation of the dataset and estimate the average recall. 
This problem is now a one-dimensional search problem: our goal is to find the smallest $m$ 
that satisfies the recall constraint.
We solve this using essentially a binary search.
If infeasible, the user relaxes one of the constraints and repeats the above method.
We experimentally validated this tuning method and compared it with an
exhaustive search on the parameter space. 
The results are discussed in Section \ref{expt:tuning}.

\section{Inverted Indexing}
\label{sec:index}

To speedup keywords and anchored regex queries on standard ASCII
text, a popular technique is to use standard
inverted-indexing~\cite{classical2}. While indexing $k$-MAP data is
pretty straightforward, the FullSFA is difficult. The reason is that the
FullSFA encodes exponentially many strings in its length, and so
indexing all strings for even a moderate-sized SFA is hopeless. Figure
\ref{infeas-direct} shows the size of the index obtained (in number of
    {\em postings}~\cite{classical2}, in our case line number item pairs)
    when we try to directly index the \name text of a single SFA (one
    OCR line).

\begin{figure}[hbtp]
\centering
\includegraphics[width=3.3in]{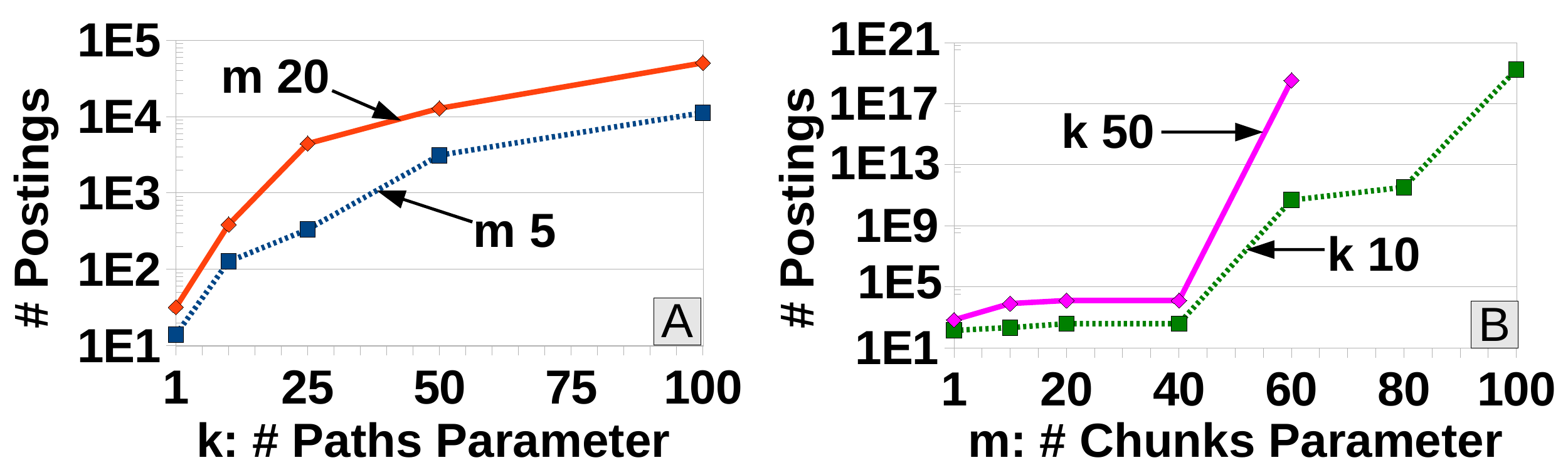}
\caption{Number of postings (in logscale) from directly indexing one SFA. 
(A) Fix $m$, vary $k$. (B) Fix $k$, vary $m$.
In (B), for $k = 50$, the number of postings overflows the 64-bit representation beyond 
$m = 60$.
% Recall that $m$ is the number of chunks parameter and $k$ is the number of paths parameter.
}
\label{infeas-direct}
\end{figure}

Figure~\ref{infeas-direct} shows an exponential blowup with $m$ --
which is not surprising as we store exponentially more paths with
increasing $m$. Our observation is that many of these exponentially
many terms are useless to applications.  Thus, to extend the reach of
indexing, we apply a standard technique. We use a dictionary of terms
input by the user, and construct the index only for these
terms~\cite{dict-technique}. These terms may be extracted from a known
clean text corpus or from other sources like an English
dictionary. Our construction algorithm builds a DFA from the
dictionary of terms, and runs a slight modification of the SFA
composition algorithm \cite{hopcroft} with the data to find the start
locations of all terms (details of the modification are in the full
version~\cite{full-paper}).  The running time of the algorithm is
linear in the size of the dictionary.
 
\paragraph*{Projection}
In traditional text processing, given the length of 
the keyword and the offset of a match, 
we can read only that small portion of the document to process the query.  
We extend this idea to \name by finding a small portion of the SFA that
is needed to answer the query -- an operation that we call {\em projection}. 
Given a term $t$ of length $u$, we obtain start
locations of $t$ from the postings. For each start
location, we compute an (over)estimate of the nodes 
that we must process to obtain the term $t$. 
More precisely, we want the descendant nodes in the DAG that
can be reached by a directed path from the start location that
contains $u$ or fewer edges (we find such nodes using a breadth-first
search). This gives us a set of nodes that we must retrieve, which is
often much smaller than the entire SFA.

We empirically show that even a simple indexing scheme
as above can be used by \name to speedup keyword and anchored
regular expression queries by over an order of magnitude versus a filescan-based
approach. This validates our claim that indexing is possible for OCR transducers, 
and opens the possibility of adapting more advanced indexing 
techniques to improve the runtime speedups.
%We demonstrate that this
%technique, helps speed up query processing for keywords and anchored
%regexes by an order of magnitude over a simple scan-based approach.

\section{Experimental Evaluation}
\label{sec:experiments}

%In this section, we present the experimental evaluation the three approaches --
% KMAP, FullFST and Staccato -- on various real world 
%data sets, comparing their performance and quality.  In this section,
We experimentally verify that the \name approach can gracefully tradeoff between
performance and quality. We also validate that our modifications to
standard inverted indexing allow us to speedup query answering.

\begin{table}[hbtp]
\centering
\begin{tabular}{|c|c|c|c|c|}
\hline
\multirow{2}{*}{\textbf{Dataset}} & \textbf{No. of} & \textbf{No. of} & \multicolumn{2}{c|}{\textbf{Size as:}} \\
 \cline{4-5}
 & \textbf{Pages} & \textbf{SFAs} & \textbf{SFAs} & \textbf{Text}\\
\hline
Cong. Acts (CA) & 38 & 1590 & 533MB & 90kB\\
\hline
English Lit. (LT) &  32 & 1211 & 524MB & 78kB\\
\hline
DB Papers (DB) & 16 & 627 & 359MB & 54kB\\
\hline
\end{tabular}
\caption{
Dataset Statistics. Each SFA represents one line of a scanned page. 
}
\label{datasets}
\end{table}

\paragraph*{Datasets Used}
We use three real-world datasets from domains where document
digitization is growing. Congress Acts (CA) is a set of scans of 
acts of the U.S. Congress, obtained from The Hathi Trust \cite{hathi}.
English Literature (LT) is a set of scans of an English literature
book, obtained from the JSTOR Archive \cite{jstor}. 
Database Papers (DB) is a set
of papers that we scanned ourselves to simulate a setting where an
organization would scan documents for in-house usage. All the scan
images were converted to SFAs using the OCRopus
tool~\cite{ocropus}. Each line of each document is represented by one
SFA.  We created a manual ground truth for these documents. The
relevant statistics of these datasets are shown in Table
\ref{datasets}. In order to study the scalability of the approaches on 
much larger datasets, we used a 100 GB dataset obtained from Google Books \cite{googlebooks}.

%We also ran the approaches on a much larger dataset
%obtained from the UNLV collection \cite{unlv}, and use a set with 8000
%SFAs (about 2.5 GB) for scalability experiments.

\paragraph*{Experimental Setup} The three approaches were
implemented in C++ using PostgreSQL 9.0.3. The current implementation
is single threaded so as to assess the impact of the approximation.
All experiments are run on Intel Core-2 E6600 machines with 2.4 GHz
CPU, 4 GB RAM, running Linux 2.6.18-194.  The runtimes are averaged
over 7 runs. The notation for the parameters is summarized in Table
\ref{notation}.

\begin{table}[hbtp]
\centering
\begin{tabular}{|c|c|}
\hline
\textbf{Symbol} & \textbf{Description}\\
\hline
$k$ & \# Paths Parameter ($k$-MAP, \name)\\
\hline
$m$ & \# Chunks Parameter (\name)\\
\hline
$NumAns$ & \# Answers queried for\\
\hline
\end{tabular}
\caption{Notations for Parameters}
\label{notation}
\end{table}

\noindent
We set $NumAns=100$, which is greater than the number of answers in
the ground truth for all reported queries. If \name finds fewer
matches than $NumAns$, it may return fewer answers.  $NumAns$ affects
precision, and we do sensitivity analysis for $NumAns$ in the full
version~\cite{full-paper}.

\subsection{Quality - Performance Tradeoff (Filescan)}
We now present the detailed quality and performance results for
queries run with a full filescan. The central technical claim of this
paper is that \name bridges the gap from the
low-recall-but-fast MAP to the
high-recall-but-slow FullSFA. 
%To verify this claim we
%run across data sets, we present Table~\ref{tab:recall-vs-runtime} XXX
%WHAT ARE OUR QUERIES? WHAT IS NUM ANSWERS XXX. This shows that indeed
\eat{To verify this claim, we measured the recall and performance of
several queries on the three datasets. Table
\ref{tab:recall-vs-runtime} presents a subset of these results 
(more queries are presented in the full version \cite{full-paper}).
We formulated these queries based on 
our discussions with researchers in the humanities and law
schools here, who use similar real-world OCR data in their research \cite{arun1}.}
To verify this claim, we measured the recall and performance of
21 queries on the three datasets. We formulated these queries based on 
our discussions with practitioners in companies and researchers in the social 
sciences who work with real-world OCR data.
Table \ref{tab:recall-vs-runtime} presents a subset of these results 
(the rest are presented in the full version of this paper \cite{full-paper}).

%% Precision results
%% \multirow{6}{40pt}{Precision} 
%% & CA1 & 1.00 & 0.28 & 1.00  \\
%% & CA2  & 1.00 & 0.55 & 1.00 \\ \cline{2-5}
%% & LT1  & 0.96 & 0.92 & 0.97 \\
%% & LT2  & 0.78 & 0.31 & 0.44 \\ \cline{2-5}
%% & DB1  & 0.93  & 0.67 & 0.90 \\
%% & DB2  & 0.96 & 0.33 & 0.91 \\
\begin{table}[hbtp]
\centering
\begin{tabular}{|c||c|c|c||c|}
%\begin{tabular}{|p{2cm}|p{0.8cm}|p{1.5cm}|p{1.5cm}|p{1.5cm}|}
\hline
 Query & MAP & $k$-MAP & FullSFA & \name \\
\hline
&\multicolumn{4}{c|}{\textbf{Precision/Recall}}\\
\hline
CA1 & 1.00/0.79    & 1.00/0.79   & 0.14/1.00 & 1.00/0.79  \\
CA2 & 1.00/0.28	   & 1.00/0.52   &  0.25/1.00  & 0.73/0.76\\
\hline
LT1 & 0.96/0.87   &  0.96/0.90   &  0.92/1.00 & 0.97/0.91 \\
LT2 & 0.78/0.66    & 0.76/0.66    &  0.31/0.97 & 0.44/0.81 \\ 
\hline
DB1 & 0.93/0.75    &   0.90/0.92  &  0.67/0.99 &  0.90/0.96 \\
DB2 & 0.96/0.76    &  0.96/0.76   &  0.33/1.00 & 0.91/0.97 \\
\hline
&\multicolumn{4}{c|}{\textbf{Runtime (in seconds)}}\\
\hline
 CA1 & 0.17 & 0.75 &  86.72&  2.87 \\
 CA2 & 0.18 & 0.84 & 150.35 & 3.36\\
\hline
 LT1 & 0.13 & 0.19  & 83.78 & 1.98 \\
 LT2 & 0.14 & 0.24   & 155.45 & 2.88 \\ 
\hline
 DB1 & 0.07 &  0.29 &  40.73 &  0.75 \\
 DB2 & 0.07 & 0.33  &  619.31 &  0.86 \\
\hline
\end{tabular}
\caption{Recall and runtime results across datasets. The keyword
  queries are -- CA1: `$President$', LT1: `$Brinkmann$' and DB1:
  `$Trio$'. The regex queries are -- CA2: 
  `$U.S.C.~2 \backslash d \backslash d \backslash d$', 
  LT2: `$19\backslash d\backslash d,~\backslash
  d\backslash d$' and DB2: `$Sec(\backslash x)*\backslash d$'. Here,
  $\backslash x$ is any character and $\backslash d$ is any digit.
  The number of ground truth matches are -- CA1: 28, LT1: 92, DB1: 68,
  CA2: 55, LT2: 32 and DB2: 33.  The parameter setting here is:
  $~k=25,~m=40,NumAns=100$.  }
\label{tab:recall-vs-runtime}
\end{table}

We classify the kinds of queries to \textit{keywords} and
\textit{regular expressions}. The intuition is that keyword queries
are likely to achieve higher recall on $k$-MAP compared to more
complex queries that contain spaces, special characters, and
wildcards.
Table \ref{tab:recall-vs-runtime} presents the recall and
runtime results for six queries -- one keyword and one regular
expression (regex) query per dataset. Table
\ref{tab:recall-vs-runtime} confirms that indeed there are
intermediate points in our approximation that have faster runtimes
than FullSFA (even 
up to
two orders of magnitude), while providing
higher quality than $k$-MAP.

We would like the tradeoff of quality for performance to be smooth as
we vary $m$ and $k$.  To validate that our approximation can support
this, we present two queries, a keyword and a regex, on the Congress
Acts dataset (described below). To demonstrate this point, we vary $k$
(the number of paths) for several values of $m$ (the number of chunks)
and plot the results in Figure~\ref{recall-runtime-vs-k}. Given an
SFA, $m$ takes values from 1 to the number of the edges in the SFA
(the latter being the nominal parameter setting `Max').  When $m=1$,
\name is equivalent to $k$-MAP.  Note that the state-of-the-art in our
comparison is essentially the MAP approach ($k$-MAP with $k=1$, or
\name with $m=1,k=1$), which is what is employed by Google Books.

\begin{figure}
\centering
\vspace{0.01in}
\includegraphics[width=3.3in]{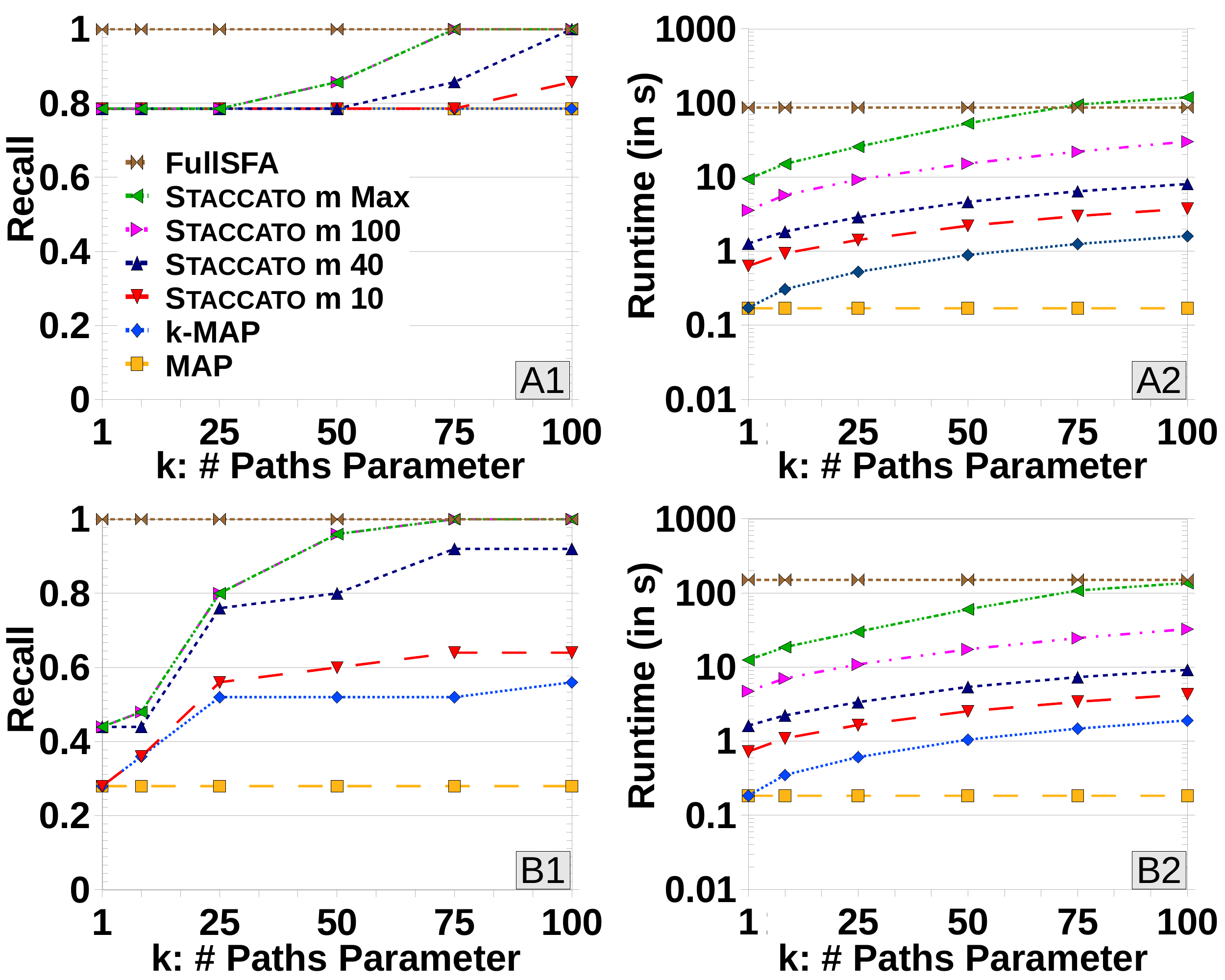}\\
\caption{Recall and Runtime variations with $k$, for different values of $m$, on two queries: (A) `$President$' (keyword), and 
(B) `$U.S.C.~2 \backslash d \backslash d \backslash d$' (regex). The $\backslash d$ is short for $(0|1|...|9)$. 
The runtimes are in logscale. $NumAns$ is set to 100. Recall that $m$ is the 
number of chunks parameter and $NumAns$ is the number of answers queried for.}
\label{recall-runtime-vs-k}
\end{figure}

\paragraph*{Keyword Queries}
In Figures \ref{recall-runtime-vs-k} (A1) and (A2), we see the recall
and performance behavior of running a keyword query (here {\it
  `President'}) in \name for various combinations of $k$ and $m$. We
observe that the recall of $k$-MAP is high ($0.8$) but not perfect and
in (A2) $k$-MAP is efficient ($0.1$s) to answer the query. Further, as
we increase $k$ there is essentially no change in recall (the running
time does increase by an order of magnitude). We verified that the
reason is that the top-$k$ paths change in only a small set of
locations -- and so no new occurrences of the string {\it
  `President'} are found. In contrast, the FullSFA approach achieves
perfect recall, but it takes over $3$ orders of magnitude longer to
process the query. As we can see from the plots, for the \name
approach, the recall improves as we increase $m$ -- with corresponding
slowdowns in query time. We believe that our approach is promising
because of the gradual tradeoff of running time for quality. The fact
that the $k$-MAP recall does not increase substantially with $k$, and
does not manage to achieve the recall of FullSFA even for large $k$
underscores the need for finer-grained partition, which is what \name
does.

\paragraph*{Regular Expressions} 
Figures \ref{recall-runtime-vs-k} (B1) and (B2) present the 
results for a more sophisticated regex query that looks
for a congressional code (`$U.S.C.~2 \backslash d \backslash d \backslash d$')
 referenced in the text. As the figure shows, this
more sophisticated query has much lower recall for the MAP
approach, and increases slowly with increasing $k$. Again, we see the same
tradeoff that the FullSFA approach is orders of magnitude slower than
$k$-MAP, but achieves perfect recall. Here, we see that the \name
approach does well: there are substantial (but smooth) jumps in quality
as we increase $k$ and $m$, going all the way from MAP to FullSFA. 
This suggests that more
sophisticated queries benefit from our scheme more, which is an encouraging 
first step to enable applications to do rich analytics over such data.
%which we believe is a promising first step toward integrating content systems with an RDBMS.

\begin{figure} [hbtp]
\centering
\includegraphics[width=3.3in]{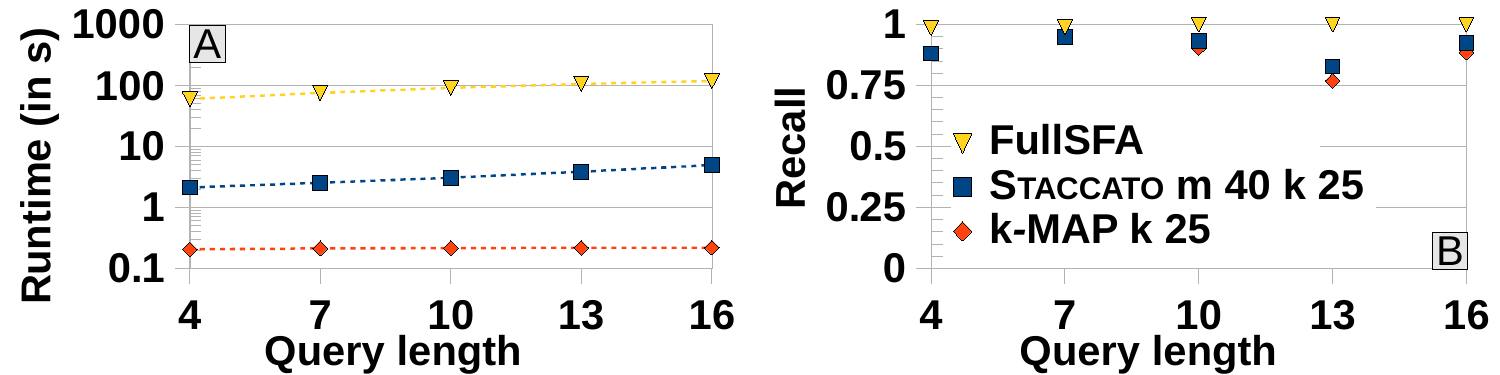}\\
\caption{Impact of Query Length on (A) Runtime and 
(B) Recall. $NumAns$, the number of answers queried for, is set to 100.
}
\label{exp:qlenshorter}
\end{figure}

\paragraph*{Query Efficiency}
To assess the impact of query length on recall and runtime, we plot the two 
for a set of keyword queries of increasing length in Figure \ref{exp:qlenshorter}.
We observe that the runtimes increase polynomially but slowly for all the approaches, 
while no clear trends exist for the recall. We saw similar results with regular
expression queries, and discuss the details in the full version~\cite{full-paper}.

We also studied the impact of $m$ and $k$ on precision (and
F-1 score), and observed that the precision of \name usually falls in between $k$-MAP
and FullSFA (but F-1 of \name can be better than both in some cases). 
Similar to the recall-runtime tradeoff, \name also manages to gracefully
tradeoff on precision and recall.
Due to space constraints, these results are discussed in the full version~\cite{full-paper}.

\subsection{Staccato Construction Time}

\begin{figure}[hbtp]
\centering
\includegraphics[width=3.3in]{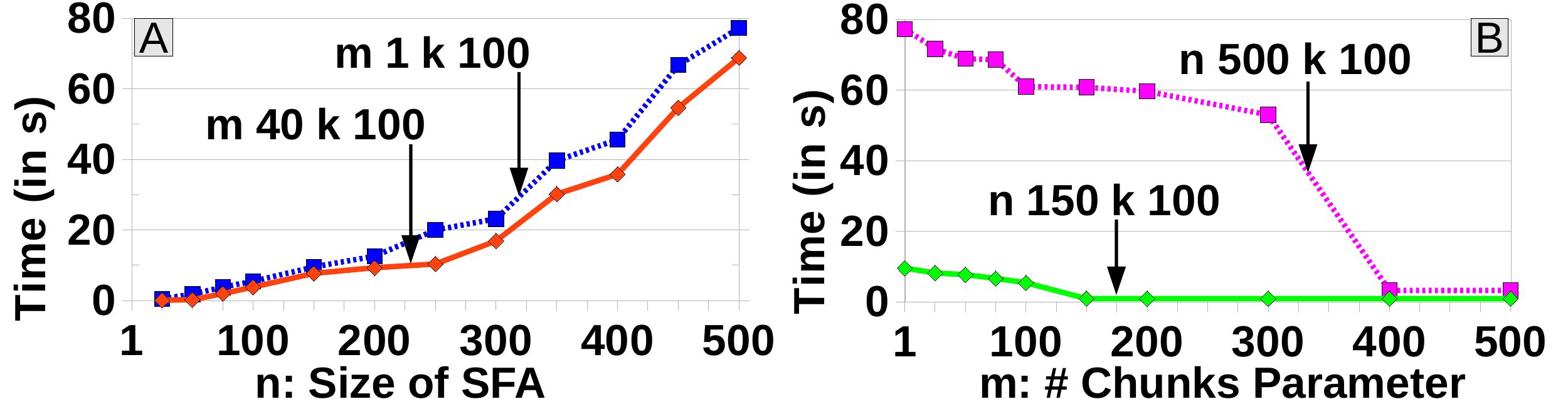}
\caption{(A) Variation of \name approximation runtimes with the size of the SFA ($n$ = number of nodes + edges) 
fixing $m$ and $k$.
(B) Sensitivity of the runtimes to $m$, fixing $n$ and $k$. 
Recall that $m$ is the number of chunks parameter and $k$ is the number of paths parameter.}
\label{stac-constr}
\end{figure}

We now investigate the runtime of the \name's approximation algorithm.
The runtime depends on the size of the input SFA data as well as $m$
and $k$.  We first fix $m$ and $k$, then we plot the construction time
for SFAs of varying size (number of nodes) from the CA dataset (Figure
\ref{stac-constr}(A)). Overall, we can see that the algorithm runs
efficiently -- even in our unoptimized implementation. As this is an
offline process, speed may not be critical for some applications. Also, this computation
is embarassingly parallel (across SFAs). We used Condor \cite{condor}
to run the \name construction on all the SFAs in the three datasets,
for all of the above parameters.  This process completed in
approximately 11 hours.

To study the sensitivity of the construction time to $m$, we select a
fixed SFA from the CA dataset (Figure \ref{stac-constr}(B)).  When $m \ge |E|$, 
the algorithm picks each transition as a block, and terminates.
But when $m = 300 < |E|$, the algorithm computes several candidate
merges, leading to a sudden spike in the runtime. There onwards, the
runtime varies almost linearly with decreasing $m$. However, there are
some spikes in the middle. We verified that the spikes arise since the
`FindMinSFA' operation has to fix merged chunks not satisfying the SFA
property, thus causing the variation to be less smooth. We also
verified that the runtime was linear in $k$, fixing the SFA and $m$
(see full version \cite{full-paper}). In general, a linear runtime in
$k$ is not guaranteed since the chunk structure obtained during
merging may not be similar across $k$, for a given SFA and $m$.
%\end{comment}

\begin{comment}
We validate that our \name construction algorithm is able to run on a
wide variety of parameters by varying the sensitivity of our
construction in Figure \ref{stac-constr} to the parameters $m$ and $k$
for a fixed SFA $(V,E,\delta)$. When $m >|E|$, the algorithm picks
each transition as a block, and terminates.  But when $m = 300 < |E|$,
the algorithm obtains many candidate merges, leading to a sudden spike
in the runtime. After this point, it is almost linear with decreasing
$m$, but there are spikes in the middle. This is due to the `FindMinSFA'
operation. The third plot shows a
simple linear trend in runtime as we vary $k$. While the plot is
smooth here, it may not be quite as smooth since different values
of $k$ may select different merging criteria.
\end{comment}

\subsection{Inverted Indexing}

We now verify that standard inverted indexing can be made to work on
SFAs. We implement the index as a relational table with a B+-tree on
top of it. More efficient inverted indexing implementations are
possible, and so our results are an upperbound on indexing
performance. However, this prototype serves to illustrate our main
technical point that indexing is possible for such data.

A dictionary of about 60,000 terms from a freely available
dictionary~\cite{dict-website} was converted to a prefix-trie
automaton, and used for index construction.  While parsing the query, we
ascertain if the given regex contains a left-anchor term. If
so, we look up the anchor in the index to obtain the postings, and retrieve
the data to employ query processing on them.

\begin{figure}[hbtp]
\centering
\includegraphics[width=3.3in]{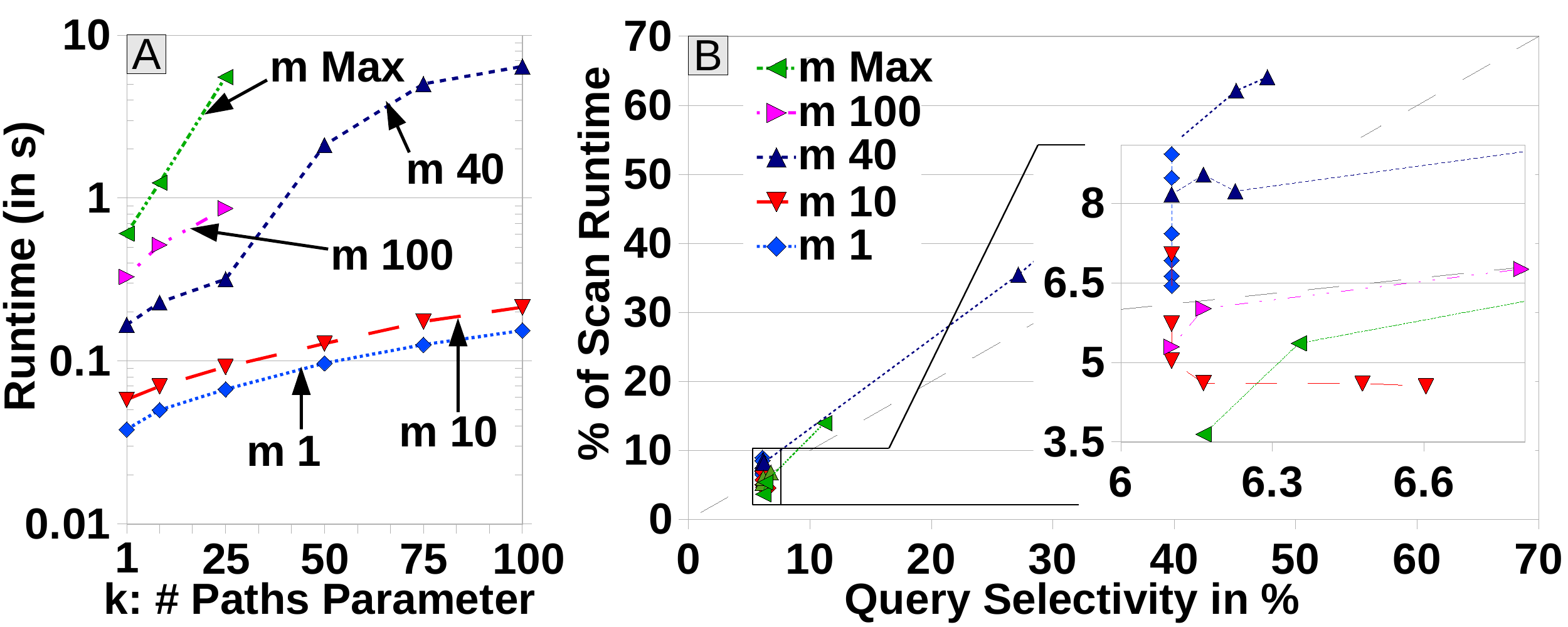}
\caption{(A) Total Runtimes, and (B) Fractional Runtimes Vs Selectivity for the query `$Public~Law~(8|9)\backslash d$', using the 
inverted index with the left anchor term `$public$'. Runtimes are in logscale. Recall that $m$ is the number 
of chunks parameter.}
\label{indexed-runtimes}
\end{figure}

Figure \ref{indexed-runtimes} shows the results for a fixed length
left anchored regex on the CA data set that is anchored by a word in
the dictionary (here, `Public'). We omit some combinations
($m=100,\mathrm{Max~and~}k=50,75,100$) since their indexes had nearly
100\% selectivity for all queries that we consider, rendering them
useless.  The first plot shows the sensitivity of the total runtimes
to $m$ and $k$. Mostly, there is a linear trend with $k$, except for a
spike at $m=40, k=50$. To understand this behavior, we plot the
runtime, as a percentage of the filescan runtime, against the
selectivity of the term in the index.  Ideally, the points should lie
on the $Y=X$ line, or slightly above it. For the lowest values of $m$
and $k$, the relative speedup is slightly lowered by the index lookup
overhead. But as $k$ increases, the query processing dominates, and
hence the speedup improves, though selectivity changes only slightly.
For higher $m$, the projection overhead lowers the speedup, and as $k$
goes up, the selectivity shoots up, increasing the runtime. Overall,
we see that dictionary-based indexing provides substantial speedups in
many cases.

%More experiments are described in the appendix.
\subsection{Scalability}

To understand the feasibility of our approaches on larger amounts of
data, we now study how the runtimes scale with increasing dataset
sizes.  We use a set of 8 scanned books from Google Books
\cite{googlebooks} and use OCRopus to obtain the SFAs. The total size
of the SFA dataset is 100 GB.

\begin{figure}[hbtp]
\centering
\includegraphics[width=3.2in]{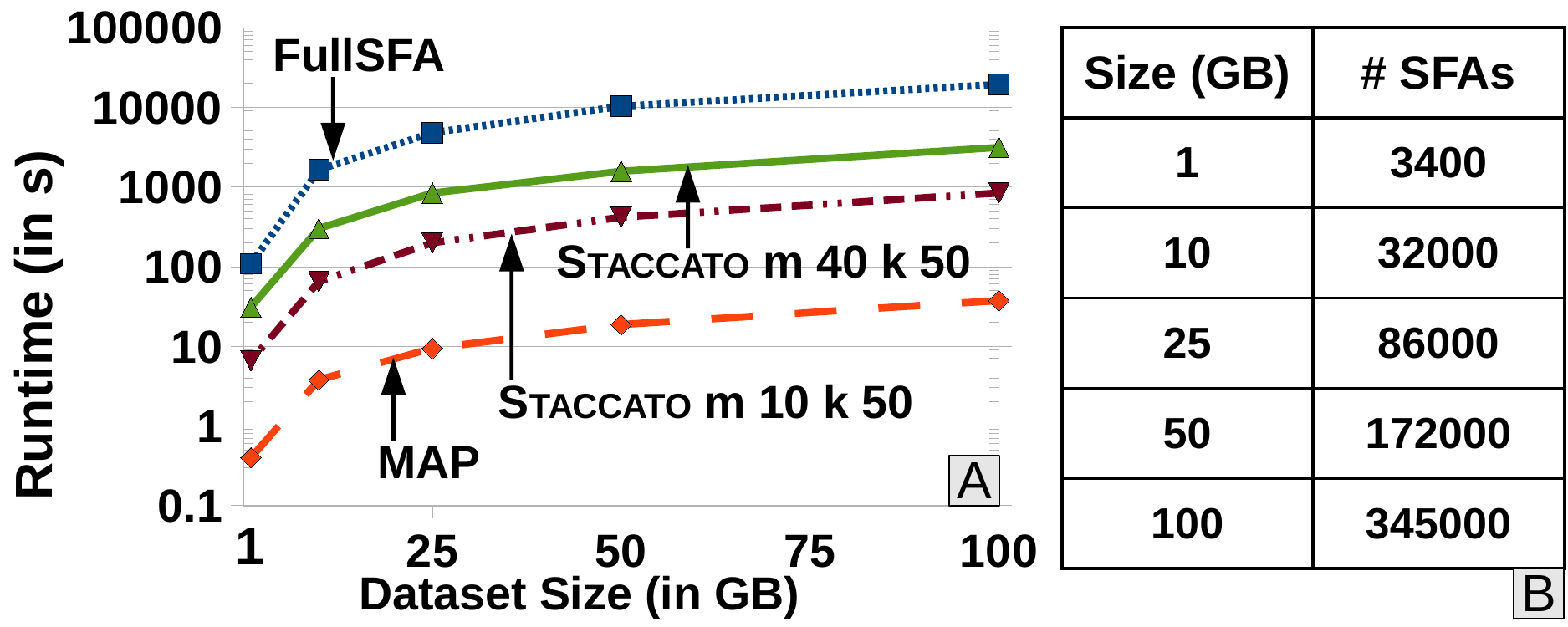}
\caption{(A) Filescan runtimes (logscale) against the dataset size for MAP, 
FullSFA and \name with two parameter settings. (B) Number of SFAs
in the respective datasets. 
%Recall that $m$ is the number of chunks parameter
%and $k$ is the number of paths parameter.
} 
\label{scalability}
\end{figure}

Figure \ref{scalability} shows the scalability results for a regex
query.  The filescans for FullSFA, MAP and \name all scale linearly in
the dataset size. Overall, the filescan runtimes are in the order a
few hours for FullSFA. The runtimes are one to two orders of magnitude
lower for \name, depending on the parameters, and about three orders
of magnitude lower for MAP. We also verified that indexing over this
data provides further speedup (subject to query selectivity) as shown
before. One can speedup query answering in all of the approaches by
partitioning the dataset across multiple machines (or even using
multiple disks). Thus, to scale to much larger corpora (say, millions
of books), we plan to investigate the use of parallel data processing
frameworks to attack this problem.

\eat{
%\includegraphics[width=3.3in]{plots/AllResults/Scalability/scalability}
%\caption{Scalability of runtimes as dataset size increases. (A) Filescan-based
%approach for MAP, FullSFA and \name. (B) \name with usage of the inverted index.
% The runtimes are in logscale.} 

%To understand the feasibility of our approaches in the real world, we
%now study how the runtimes scale with larger dataset sizes. 
To understand the feasibility of our approaches on larger amounts of data, we
now study how the runtimes scale with increasing dataset sizes. 
For this study, we used a collection of about 200 pages from the UNLV dataset
\cite{unlv}. Figure \ref{scalability} shows the results for an
anchored regex query.  Clearly, the filescans for FullSFA, MAP and
\name are all are linear in the dataset size. Overall, the filescan
runtimes are in the order of a few minutes for the entire dataset.
Also, the indexed access for \name with $m=10 ,k=100$ gives an order of
magnitude speedup, subject to a selectivity of about 5-7\% in the
index here across the dataset sizes.  However, the indexed access for
$m=40,k=50$ does not help much, since the selectivity of the
anchor term in the index here was 58-69\%.  For much larger datasets
(millions of documents), even the simple MAP approach cannot scale on
one machine. As our goal is to eventually index huge corpora, we plan
to investigate the use of parallel data processing frameworks to
attack this problem.
}

\subsection{Automated Parameter Tuning}
\label{expt:tuning}
We now empirically demonstrate the parameter tuning method on a labeled 
set of 1590 SFAs (from the CA dataset), and a set of 5 queries (both keywords and regular 
expressions). The size constraint is chosen as 10\% and the recall constraint 
is chosen as 0.9. We use increments of 5 
for both $m$ and $k$. Based on the tuning method described in Section \ref{sec:analysis}, we obtain the following size 
equation: $20mk + 58k = 45540$, and the resultant parameter estimates of 
$m = 45,~ k = 45$, with a recall of 0.91. We then performed an exhaustive search 
on the parameter space to obtain the optimal values subject to the same constraints.
Figure \ref{fig:exhaustive} shows the surface plots 
of the size and the recall obtained by varying $m$ and $k$. The optimal 
values obtained are: $m = 35,~ k = 80$, again with a recall of 0.91. 
The difference in the parameter values arises primarily because the tuning method
overestimated the size at this location. Nevertheless, we see that the tuning method 
provides parameter estimates satisfying the user requirements.

\begin{figure}[hbtp]
\centering
\vspace{0.01in}
\includegraphics[width=3.3in]{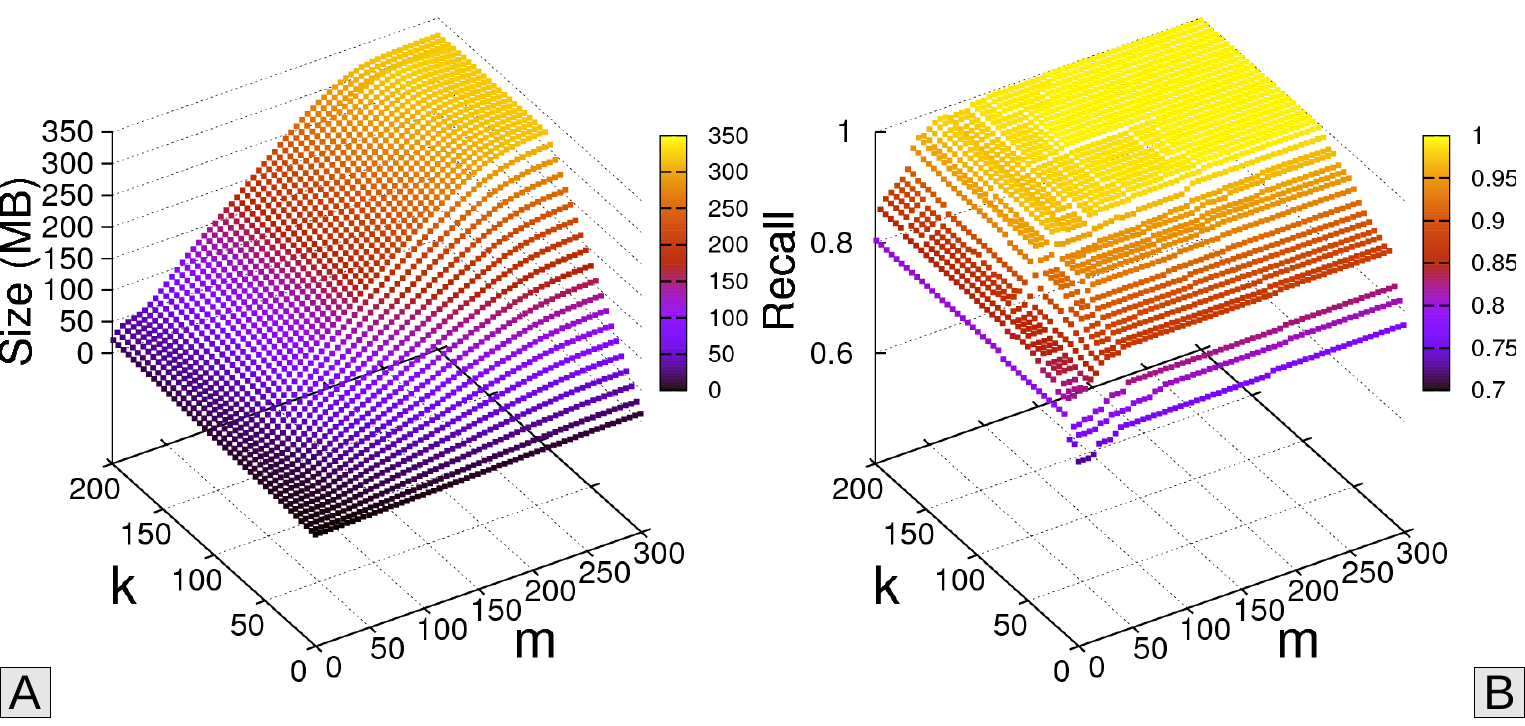}
\caption{
3-D plots showing the variation of (A) the 
size of the approximated dataset (in MB), and 
(B) the average recall obtained. Recall that $m$ is the number of chunks 
parameter and $k$ is the number of paths parameter.
}
\label{fig:exhaustive}
\end{figure}

\section{Related Work}
\label{sec:related_work}

Transducers are widely used in the OCR and speech
communities~\cite{mohri-transducers, mohri-composition} and mature
open-source tools exist to process in-memory transducers
\cite{mohri-openfst}. For example we use a popular open-source tool,
OCRopus \cite{ocropus}, from Google Books that provides well-trained
language models and outputs transducers. See
Mohri~et~al.~\cite{mohri-transducers} for a discussion of why
transducers are well-suited to represent the uncertainty for OCR. In
the same work, Mohri~et~al. also describe speech data. We experimented
with speech data, but we were hampered by the lack of high quality
open-source speech recognizer toolkits. Using the available toolkits,
we found that the language quality from open source speech recognizers
is substantially below commercial quality.

The Lahar system \cite{re-query, julie-approx} manages Hidden Markov
Models (HMMs) as \textit{Markovian streams} inside an RDBMS and allows
querying them with SQL-like semantics.  In contrast to an HMM
\cite{hmm} that requires that all strings be of the same length,
transducers are able to encode strings of different lengths. This is
useful in OCR, since identifying spaces between words is difficult,
and this uncertainty is captured by the branching in the
SFA~\cite{mohri-transducers}. Our work drew inspiration from the
empirical study of work of approximation trade-offs from Letchner et
al.~\cite{julie-approx}. Directly relevant to this work is the recent
theoretical results of Kimelfeld and R{\'e} \cite{re-transducers}, who
studied the problem of evaluating transducers as queries over
uncertain sequence data modeled using Hidden Markov Models
\cite{re-query, hmm}. \name represents both the data and query by
transducers which simplifies the engineering of our system.

Transducers are a graphical representation of probability models which
makes them related to graphical models. Graphical models have been a
hot topic in the database research community. Kanagal et
al.~\cite{amol-query} handle general graphical models. Wang et
al.~\cite{daisy-bayes} also process Conditional Random Fields (CRFs)
\cite{crf}. Though transducers can be viewed as a specialized directed
graphical model, the primary focus of our work here is on the
application of transducers to OCR in the domain of content management
and the approximations that are critical to achieve good performance.
However, our work is similar in spirit to these in that we too want to
enable SQL-like querying of probabilistic OCR data inside an RDBMS.

Probabilistic graphical models have been successfully applied to
various kinds of sequential data including OCR \cite{hmm-ocr}, 
RFID \cite{re-query}, speech \cite{hmm-speech}, 
%bio-sequence analysis \cite{book-hmm-bio}, 
etc.  Various models have been studied in both
the machine learning and data management communities \cite{book-cowell, book-learning,
  amol-query, re-query, daisy-bayes}.%, yanlei-model}.

Many approximation schemes for probabilistic models have been studied
~\cite{approx-ai, julie-approx}. We built on the technique
$k$-MAP~\cite{murphy}, which is particularly relevant to us.
Essentially, the idea is to infer the top $k$ most likely results from
the model and keep only those around. Another popular type of
approximation is based on {\em mean-field theory}, where the intuition
is that we replace complex dependencies (say in a graphical model)
with their average (in some sense)~\cite{wainwright}. Both mean-field
theory and our approach share a common formal framework: minimizing
KL-divergence. For a good overview of various probabilistic graphical
models, approximation and inference techniques, we refer the reader to
the excellent book by Wainwright and Jordan~\cite{wainwright}.

Gupta and Sarawagi~\cite{guptasarawagi} devise efficient 
approximation schemes to represent the outputs of a CRF, viz., labeled 
segmentations of text, in a probabilistic database. They partition
the space of segmentations (i.e., the outputs) using boolean constraints on the output 
segment labels, and then structurally merge the partitions to a pre-defined count using 
Expectation Maximization, without any enumeration. Thus, their final partitions are
%and thus obtain partitions that are
disjoint sets of full-row outputs (`horizontally' partitioned). 
Both their approach and \name use 
KL-divergence to measure the goodness of approximation. 
However, \name is
different in that we partition the underlying structure of the
model (`vertically' partitioned). 
They also consider soft-partitioning approaches to overcome
the limitations of disjoint partitioning. It is interesting future work to adapt such 
ideas for our problem, and compare with \name's approach.

Probabilistic databases have been studied in several recent projects
(e.g., ORION~\cite{orion},Trio~\cite{anish-icde},
%ConQuer~\cite{conquer}, 
MystiQ~\cite{DBLP:conf/vldb/DalviS04},
Sprout~\cite{DBLP:conf/icde/OlteanuHK09}, and
MayBMS~\cite{olteanu-maybms}). Our work is complementary to these efforts: the queries
we consider can produce probabilistic data that can be ingested by
many of the above systems, while the above systems focus on querying
restricted models (e.g., U-Relations or BIDs). We also use model-based
views \cite{amol-mauvedb} to expose the results of query-time
inference over the OCR transducers to applications.

The OCR, speech and IR communities have explored error correction
techniques as well as approximate retrieval schemes \cite{inquery,
  ocr-ngrams, ocr-book}. However, prior work primarily
focus on keyword search over plain-text transcriptions.  \name can
benefit from these approaches and is orthogonal to our goal of
integrating OCR data into an RDBMS. In contrast, we advocate retaining
the uncertainty in the transcription.

Many authors have explored indexing techniques for probabilistic data
\cite{julie-index, amol-index, upi-index, singh-index}.  Letchner et al.~\cite{julie-index} design new indexes
for RFID data stored in an RDBMS as Markovian streams. Kanagal et
al.~\cite{amol-index} consider indexing correlated probabilistic
streams using tree partitioning algorithms and describe a new
technique called \textit{shortcut potentials} to speedup query
answering.  Kimura et al.~\cite{upi-index} propose a new
\textit{uncertain primary index} that clusters heap files according to
uncertain attributes. 
%Sarma et al.~\cite{anish-index} study indexing
%and usage of statistics over uncertain data stored as relations. 
Singh et al.~\cite{singh-index} consider indexing categorical data and
propose an R-tree based index as well as a probabilistic inverted
index. Our work focuses on the challenges that content models like OCR
raise for integrating indexing with an RDBMS.

\section{Conclusion and Future Work}

We present our prototype system, \name, that integrates
a probabilistic model for OCR into an RDBMS. 
We demonstrated that it is possible to
devise an approximation scheme that trades query runtime performance
for result quality (in particular, increased recall). The technical
contributions are a novel approximation scheme and a formal analysis
of this scheme. Additionally, we showed how to adapt standard
text-indexing schemes to OCR data, while retaining more answers.

Our future work is in two main directions. Firstly, we aim to
extend \name to handle larger data sets and more sophisticated
querying (e.g., using aggregation with a probabilistic RDBMS, 
sophisticated indexing, parallel processing etc.). 
Secondly, we aim to extend our techniques to more
types of content-management data such as speech
transcription data. Interestingly, transducers provide a unifying
formal framework for both transcription processes. Our initial
experiments with speech data suggest that similar approximations
techniques may be useful. This direction is particularly exciting to
us: it is a first step towards unifying RDBMS and content-management
systems, two multibillion dollar industries.

\submissionversion{\balance}

\vspace{0.25in}

\noindent
{\Large \bf Acknowledgments} \\

\noindent 
This work is supported by the National Science Foundation 
under IIS-1054009, the Office of Naval Research under N000141210041
and the Microsoft Jim Gray Systems Lab. 
Any opinions, findings, conclusions or recommendations 
expressed in this work are those of the authors and 
do not necessarily reflect the views of 
the US government or Microsoft. 
The authors also thank the anonymous PVLDB reviewers as well
as Jignesh Patel and Benny Kimelfeld
for their valuable feedback on an earlier version of this paper.

%\bibliography{}
%{\scriptsize \bibliographystyle{abbrv} \bibliography{hazy.ocr.bib} }
\bibliographystyle{plain} 
%\bibliography{hazy.ocr.bib} 
\bibliography{HazyOCR_VLDB2012_full.bib}

\appendix
\section{Finite State Transducers}
As mentioned in Section \ref{sec:prelim}, the formal model used by \name 
to encode the uncertainty information in OCR data is the Finite
State Transducer (FST). A transducer is an automaton that converts 
(\textit{tranduces}) strings from an input alphabet to an output alphabet. 
We can view a transducer as an SFA that both reads and emits characters 
on its transitions. 
Formally, we fix an input alphabet $\Gamma$ and an output alphabet $\Sigma$.
An FST $S$ over $\Gamma$ and $\Sigma$ is a tuple $S=(V,E,s,
f,\delta)$ where $V$ is a set of nodes, $E \subseteq V \times V$ is a
set of edges such that $(V,E)$ is a directed acyclic graph, and $s$
(resp. $f$) is a distinguished start (resp. final) node (state). Each
edge has finitely many arcs. The function $\delta$ is a stochastic
transition function, i.e.,
\[ \delta : E \times \Gamma \times \Sigma \to [0,1] \text{ s.t.} \sum_{\stackrel{y: (x,y) \in E}{\gamma \in \Gamma,\sigma \in \Sigma}} 
\delta((x,y),\gamma,\sigma) = 1 \quad  \forall x \in V \]
In essence, $\delta( e, \gamma, \sigma)$, where $e = (x,y)$, is the conditional
probability of transitioning from $x \to y$, reading $\gamma$ and emitting $\sigma$. 
In OCR, the input alphabet is an encoding of the
location of the character glyphs in the image, while the output alphabet is the 
set of ASCII characters.
An FST also defines a discrete probability distribution over 
strings through its outputs. 

\section{Illustrations for FindMinSFA}
We now present more illustrations for the FindMinSFA operation (Section 
\ref{sec:approximation}) in Figure \ref{fig:findminsfamore}. As shown in Algorithm \ref{alg:findminsfa}, 
three cases arise when the given subset of nodes of the SFA $S$ do not
form an SFA by themselves. Firstly, they might not have a unique start node, in which
case their least common ancestor has to be computed (Figure \ref{fig:findminsfamore} (A)).
Secondly, they might not have a unique end node, in which case their greatest common
descendant has to be computed (Figure \ref{fig:findminsfamore} (B). Finally, there could
be an external edge incident upon an internal node of the subset (Figure \ref{fig:findminsfamore} (C)).
In all cases, FindMinSFA outputs a subset of nodes that form a valid SFA, which is then collapsed and
replaced with a single edge in $S$.

\begin{figure*}[ht]
\centering
\includegraphics[width=6.5in]{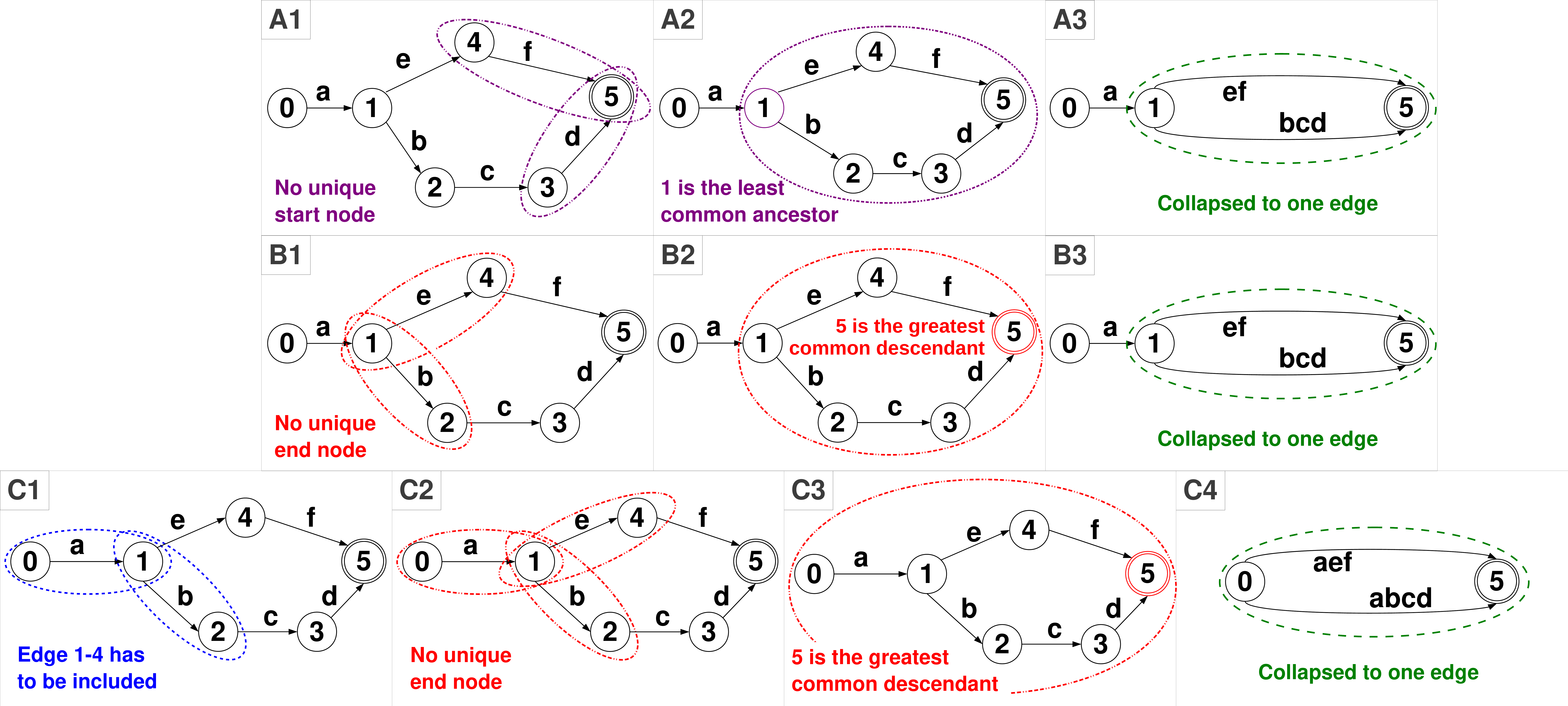}
\caption{
Illustrating FindMinSFA: (A) No unique start node for the set $X=\{3,4,5\}$, 
(B) No unique end node for the set $X=\{1,2,4\}$, and 
(C) Set $X=\{0,1,2\}$ has external edge $1-4$ incident on internal node $1$, and has to be included.
}
\label{fig:findminsfamore}
\end{figure*}

\section{Conditional is a $\mathsf{KL}$ Minimizer}
\label{app:kl}

$\mathsf{KL}$ divergence is similar
  to a distance metric in that it allows us to say whether or not two
  probability distributions are close. Given two probability
  distribution $\mu,\nu : \Sigma^{*} \to [0,1]$ the $\KL$-divergence
  is denoted $\KL(\mu || \nu)$ and is defined as:
\[ \KL(\mu || \nu) = \sum_{\sigma \in \Sigma^{*}} \mu(x) \log \frac{\mu(x)}{\nu(x)} \]
The above quantity is only defined for $\mu,\nu$ such that $\mu(x) >
0$ implies that $\nu(x) > 0$. If $\mu = \nu$ then $\KL(\mu || \nu) =
0$.

We justify our choice to retain the probability of each string we
select by showing that it is in fact a minimizer for a common
information theoretic measure, $\mathsf{KL}$-divergence. Given a
probability distribution $\mu$ on $\Sigma^{*}$ and a set $X \subseteq
\Sigma^{*}$, let $\mu_{|X}$ denote the result of conditioning $\mu$ on
$X$. Let $A$ be the set of all distributions on $X$. Then,
\begin{equation}
 \KL( \mu_{|X} || \mu) \leq \min_{\alpha \in A} \KL( \alpha ||
\mu) 
\label{eq:propto}
\end{equation}
That is, selecting the probabilities according to the conditional
probability distribution is optimal with respect to $\KL$
divergence. Eq.~\ref{eq:propto} follows from the observation that
$\KL(\mu_{|X} || \mu) = - \log Z$ where $Z = \sum_{x \in X}
\mu(x)$. Using the log-sum inequality one has
\[ \sum_{x \in X} \alpha(x) \log \frac{\alpha(x)}{\mu(x)} \geq 
\left(\sum_{x \in X} \alpha(x)\right) \log \frac{ \left(\sum_{x \in X}
    \alpha(x) \right)}{ \left(\sum_{x \in X} \mu(x) \right) }
= - \log Z \]
%}

%\section{$\delta_k$ is a minimizer (Proposition~\protect\ref{prop:min})}
\section{$\delta_k$ is a minimizer (Proposition 3.1)}

There are two observations. The first is that by normalization, since
the probability of every string is simply proportional to its
probability in $\Pr_{[\delta]}$ then the $\mathsf{KL}$ divergence is
  inversely proportional to the probability mass retained. Thus, the
  minimizer must retain as much probability mass as possible. The
  second observation is the following: consider any chunk $(S_i,s,f)$
  where $s$ is the single start state and $f$ is the final state. By
  construction, every path that uses a character from $S_i$ must
  enter through $s$ and leave through $f$. And the higher probability
  that we place in that state, the higher the retained mass. Since
  $\delta_k$ retains the highest probability in each segment, it is
  indeed the minimizer.

\section{Proof of Theorem~\protect\ref{THM:HARD}}
 
The starting point is that the following problem is \NP-hard: Given 
vectors $x,y \in \mathbb{Q}^{l}$ and a fixed constant $\lambda \geq 0$
for $l=4$ and sets of stochastic matrices
$\mathcal{S}_1,\dots,\mathcal{S}_N$ where each $\mathcal{S}_i$ is a
set of $2$ $l \times l$ matrices, determine if there is a sequence
$\bar i \in \set{1,2}^{N}$ such that $M_{i_j} \in \mathcal{S}_j$ and:
\[ x^{T} M_{i_N} \cdots M_{i_1} y \geq \lambda \]

We find a small $l$ such that the claim holds. For this, we start with
the results of Bournez and Branicky who show that a related problem called
the {\em Matrix Mortality problem} is \NP-hard for matrices of size $2
\times 2$~\cite{DBLP:journals/mst/BournezB02}, where we ask for to
find a selection as above where $x^T M_{i_N} \cdots M_{i_1} y =
0$. Unfortunately, the matrices $(M_{ij})$ are not stochastic (not
even positive). However, using the techniques of Turakainen
\cite{turakainen:1969} (and
Blondel~\cite{DBLP:journals/ipl/BlondelT97}), we can transform the
matrices into slightly larger, but still constant dimensions,
stochastic matrices ($l =4$).

Now, we construct a transducer and chunk structure $\Phi$ such that if it is possible to choose at most $k=2$ in each
chunk with the total probability mass being greater than $\lambda 2^{-N}$, 
then we can get a choice for $\bar i$. Equally, if there exists such a choice 
for $\bar i$, then we can find such a transducer representation.
So the problem of
finding the highest mass representation is \NP-hard as well.

\eat{\begin{figure}
\centering
\includegraphics[width=3in]{images/gadgets3}
\caption{Overall Reduction. We create one chunk for each of the $N$
  sets $\mathcal{S}_i$.}
\label{fig:overall:gadget}
\end{figure}
}
Throughout this reduction, we assume that every edge is assigned a
unique character to ensure the unique path property. It is
straightforward to optimize for a binary alphabet: simply add replace
each character a sequence of edges with a binary encoding (then make
this one chunk). So, we omit
the emitted string in the transition function.

Let $P(x)$ denote the probability mass that a string is emitted that
passes through the node $x$. We will group nodes together as
components of a vector. The start node $s$ has $P(s) = 1$. Then, we
construct the vector $y$ by creating nodes $v_1,\dots,v_l$ with a
transition $\delta((s,v_i), 0) = y_i$. Thus, $P(v_i) = y_i$. We need two
main gadgets: (1) A gadget to encode matrix multiplication and (2) a
gadget that intuitively encodes that given two inputs, we can select
one or the other, but not both. We provide a slightly weaker property:
For a fixed parameter $\alpha \geq 0$ (we pick $\alpha$ below). We
construct a gadget that takes as input two nodes $x, x'$ and has
output two nodes $u, u'$ such that the probability at $u$ ($P(u)$) and
at $u'$ ($P(u')$) satisfies the following weak-exclusivity:
\begin{equation}
(P(u),P(u')) = \set{(P(x),0), (0,P(x')), (v,v')} 
\label{eq:exclusive}
\end{equation}
\noindent where $v \leq \alpha P(x)$ and $v' \leq \alpha
P(x')$. Intuitively, this gadget forces us to choose $x$ or $x'$ or
not both. Notice that if we select both, then for sure the output of
each component is smaller than $\alpha$. 

\begin{figure*}[hbtp]
\centering
\includegraphics[width=6.5in]{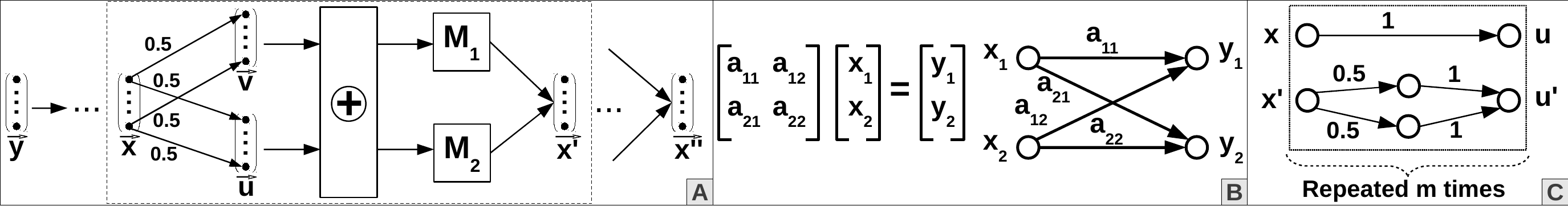}
\caption{Gadgets used in the proof of Theorem 3.1: (A) Overall Reduction. 
We create one block for each of the $N$ sets $\mathcal{S}_i$
(B) Multiply Gadget (C) Binary Exclusive Gadget
}
\label{fig:gadgets}
\end{figure*}

Assuming these gadgets, the overall construction is shown in
Figure~\ref{fig:gadgets}(A) that illustrates a chunk for a single
set $\mathcal{S}_i$. Each chunk contains two matrix multiply gadgets
(representing the two matrices in $\mathcal{S}_i$) and a large gadget
that ensures we either choose one matrix or the other -- but not
elements of both.  The input to the chunk is a vector $x$: in the
first chunk, $x = y$ above. In chunk $j$, $x$ will represent the
result of some of choice $M_{i_{j-1}} \cdots M_{i_{1}} y$. As shown,
for each $i=1,\dots,l$, we send $x_i$ to $v_i$ with probability $0.5$
and $x_i$ to $u_i$ with probability $0.5$. In turn, $u$ is fed to the
multiply gadget for $M_{1j}$ and $v$ is fed to the multiply gadget for
$M_{2j}$ (with $\delta$ = 1). We ensure that we cannot select both
$v_i$ and $u_j$ for any $i,j$ using the exclusive gadget described
below. The output of this chunk is either $0.5M_{i1}x$ or $0.5
M_{i2}x$ or a vector with $\ell_1$ norm smaller than $\alpha$. We set
$\alpha < 2^{-N}\lambda$. Given this, property is clear that given any
solution to the original problem, we can create a solution to this
problem. On the other hand, if the solution with highest probability
mass has mass greater than $2^{-N} \lambda$ then it must be a valid
solution (since we set $\alpha < 2^{-N}\lambda$). Now the gadgets:

\eat{\begin{figure}
\includegraphics[width=3in]{images/gadgets1}
\caption{Multiply Gadget}
\label{fig:multiply}
\end{figure}
}
\paragraph*{The Multiply Gadget}
Matrix multiplication can be encoded via a transducer (see
Fig.~\ref{fig:gadgets}(B)). Notice that the ``outputs'' in the above
gadget have the probability of the matrix multiply for $2\times 2$
matrices. That is, given a matrix $A$ and input nodes $x_1,\dots,x_m$,
the output nodes $y_i$ above are such that $P(y_i) = \sum_{j=0}^{m}
A_{ij} P(y_j)$. Each edge is a single chunk.

\eat{\begin{figure}
\includegraphics[width=3in]{images/gadgets2}
\caption{Binary Exclusive Gadget}
\label{fig:multiply}
\end{figure}
}

\paragraph*{The Exclusive Gadget}
We illustrate the gadget for $k=2$. We have two inputs $x,x'$ and
two outputs $u,u'$. Our goal is to ensure the property described by
Eq.~\ref{eq:exclusive}. The gadget is shown in
Figure~\ref{fig:gadgets}(C). The chunk here contains the entire
gadget, since we can only select $k$ paths, it is clear that each
iteration of the gadget we get the property that:
\[(P(u),P(u')) = \set{(P(x),0), (0,P(x')), (v,v')} \]
where $v \leq 0.5P(x)$ and $v' \leq 0.5 P(x')$. Repeating the gadget
$m$ times (taking $m$ s.t. $2^{-m} \leq 2^{-N}\lambda$ suffices). The property we need
is that we only select one vector or the other -- to ensure this we
simply place ${l \choose 2}$ gadgets: each one says that if $u_i > 0
\implies v_j = 0$ (and vice versa). Observe that the resulting gadget
is polynomial sized. After concatenating the gadgets together, the end
result is either a $2^{-N} M_{i_{N}} \cdots M_{i_i}$ (a valid result)
or its $\ell_1$ norm is smaller than $\lambda 2^{-N}$ (since it messes up
on at least one of the gadgets).

\paragraph*{Proof of Hardness} 
We now complete the proof by showing that a problem, called
\textsc{StocAut} is \NP-hard -- for a fixed size alphabet. Then, we
show how to encode this in the \textsc{Layout} problem, thereby
proving that \textsc{Layout} is \NP-hard for a fixed size $\Sigma$.

The \textsc{GenAut} Generic automaton problem is the following: Given
vectors $x,y \in \mathbb{Q}^{k}$ for some fixed $k$ and sets of
matrices $\mathcal{G}_1,\dots,\mathcal{G}_N$ where each
$\mathcal{G}_i$ for $i=1,\dots,N$ is a set of $k\times k$ matrices,
the goal is to determine if there is a tuple of natural numbers $\bar
i$ such that:
\[ x^{T} M_{i_N} \cdots M_{i_1} y \geq 0 \text{ where } M_{i_j} \in \mathcal{G}_j \]

We are concerned with the related problem \textsc{StocAut} where all
matrices and vectors in the problem are stochastic. A matrix is
stochastic if all entries are positive, and its row sums and column sums are
$1$. A vector is stochastic if each entry is positive and its $1$-norm
(sum of entries) is $1$. In \textsc{StocAut} the condition we want to
check is slightly generalized:

\[ x^{T} M_{i_N} \cdots M_{i_1} y \geq k^{-1} \text{ where } M_{i_j} \in \mathcal{G}_j \]
and $k$ is the dimension of the problem.

\begin{lemma}
For fixed dimension, $k=2$, the \textsc{GenAut} problem is
\NP-Complete in $N$ even if $|\mathcal{G}_i| \leq 2$ for $i=1,\dots,N$.
\end{lemma}

\begin{proof}
This is shown by using the observation that $x^{T} \bar A y = 0$ is
\NP-hard (exact sequence mortality
problem~\cite{DBLP:journals/mst/BournezB02}) and that $x^{T} \bar A y
= 0$ if and only if $- x^{T} \bar A y x^{T} \bar A y \geq 0$. We then
observe that $yx^{T}$ is such a matrix. More precisely, the following
problem is \NP-hard: Given a set of values $s_1, \dots, s_N$ is there
a set $S \subseteq [N]$ such that $\prod_{i=1,\dots,N} s_i = b$ for some
fixed $b$. 
\[ H = \left(\begin{array}{cc} 0 & 1\\0&-1\\ 
\end{array}\right)
\]
It is not hard to check that for any $2 \times 2$ $A$ we have:
\[ HAH = 
\left(\begin{array}{cc}
0 & (a_{21} - a_{22})\\
0 & -(a_{21} - a_{22})\end{array}\right)
%\end{array}\right)
\]
so that $HAH = 0$ if and only if $a_{21} = a_{22}$. Then, we create
the following matrices:
\[ S_i = \left(\begin{array}{cc} 1& 0\\0&s(A)\end{array}\right) \text{ and } B = \left(\begin{array}{cc} 1& 0\\b&1\end{array}\right) \]
Then, denote by $I_2$ the identity matrix. Then, we set: $\mathcal{G}_0 = \set{HB}$
and $\mathcal{G}_i = \set{S_i,I_2}$ for $i=1,\dots, N$. The vectors are $x =
(1,1)$ and $y = x^{T}$. Then, applying the construction above proves
the claim.
\end{proof}

We now apply Turakainen's technique~\cite{turakainen:1969} (we learned
of the technique from Blondel~\cite{DBLP:journals/ipl/BlondelT97}) to
transform the above matrices into slightly larger, but still constant
dimensions, that are positive and then finally stochastic. First, we
define a further restriction \textsc{ZeroAut} which requires that each
matrix row/column sum is zero. We prove that this is still
\NP-complete over slightly larger matrices.

\begin{lemma}
For dimension $k=4$, the \textsc{ZeroAut} problem is \NP-complete.
\end{lemma}

\begin{proof}
For each matrix $M$ above we create a new $4 \times 4$ matrix $N$ as follows:

\[ N = \left(\begin{array}{ccc} 
   0 & \bar 0  & 0 \\
   \alpha(x) & M & 0 \\
   \beta_0(x) & \beta(x)& 0 \\
\end{array}\right) \]
we choose $\alpha,\beta$ to be vectors such that the row,column sums
are zero. Then, we take the original $x$ and create $x_1 = (0,x,0)$
and $y_1 = (0,y^{T},0)^{T}$. It follows then that for any set of
matrices:
\[ x_1 \bar N y_1^{T} = x M y^{T} \]
Thus, we have shown the stronger claim that all products are equal and
so it directly follows that the corresponding decision problem is
\NP-complete.
\end{proof}

We define \textsc{StocAut} to be the restriction that all matrices are stochastic (as
above) and we check a slightly generalized condition:
\[ x^{T} M_{i_N} \cdots M_{i_1} y \geq k^{-1} \]
where $M_{i_j} \in \mathbb{R}^{k}$.\footnote{Although unnecessary for
  our purpose we observe that with $r$ repetitions of the problem, one
  can drive the constant down to $k^{-r}$. Thus, the constant here is
  arbitrary.}
\begin{lemma}
The problem \textsc{StocAut} is \NP-complete for matrices of $4 \times 4$.
\end{lemma}
\begin{proof}
We first show that we may assume that the vectors from the
\textsc{ZeroAut} problem are stochastic. Let $\mathbf{1}$ be the all
ones vector in $\mathbb{R}^k$. Observe that for any $Z$ such that the
row and column sums are $0$ (as is each matrix in the input of \textsc{ZeroAut})
then:
\[ (x + \alpha \mathbf{1})^T{Z}(y + \alpha \mathbf{1}) = x^T Z y \]
for any value of $\alpha$ since $\mathbf{1}^TZ = Z\mathbf{1} =
\mathbf{0}$ the zero vector. Take $\alpha = |\min_{i}
\min_{z=x,y}z_i|$. Thus, we can take $x + \alpha \mathbf{1}$ and
$y+\alpha \mathbf{1}$ and preserve the product (which are both
entry-wise positive). Then, we can scale both $x + \alpha \mathbf{1}$
and $y + \mathbf{1}$ by any positive constant:
\[ (\alpha x)^{T} Q (\beta y) = \alpha \beta x^{T} Q y \]
So the sign is preserved if $\alpha, \beta > 0$. And we can assume
without loss of generality that both $x,y$ are stochastic by scaling
by their respective $\ell_1$ norms. (If either $\norm{x}_1 = 0$ or
$\norm{y}_1 = 0$ the problem is trivially satisfied: the product is
$0$ and since they must be the zero vector.)

We now show how to achieve the condition of stochastic
matrices. First, we show how to achieve positive matrices. To do this,
let $Q$ be the all ones matrix of $4 \times 4$ ($Q_{ij} = 1$). Then
let $\lambda \geq 0$ be such that for all matrices $M + \lambda Q \geq 0$
entrywise. Then, for each matrix $M$ replace it with a new matrix $N$ defined as:
\[ N = (\lambda k)^{-1}\left( M + \lambda Q \right) \]
Since $M Q = 0$, the following holds:
\[ x^{T} \prod_{i=1}^{N} (M_i + \lambda Q) y = 
(\lambda k)^{-m} x^{T} \prod_{i=1}^{N} M_iy + x^{t} Q y \] Thus, the
original product is positive if and only if the modified product is
bigger than $k^{-1}$ proving the claim.
\end{proof}

\section{Algorithms}
\begin{figure}[hbtp]
\textbf{Notations:}\\ 
$G = (V, E, s, f, \delta)$, the \name data SFA\\ 
$Q = (V_Q, E_Q, s_Q, F_Q, \delta_Q)$, the dictionary DFA\\ 
$D = \{(f_Q \in F_Q,term)\}$, the dictionary terms (hash table)\\
$AugSts = \{(v_Q \in V_Q, PostingSet)\}$, augmented states (hash table)\\
$I = \{(term, PostingSet)\}$, the index (hash table)\\
\caption{Notations for Algorithms \ref{alg:indexing-staccato-dyn}
and \ref{alg:indexing-staccato}}
\label{alg:indexing-notations}
\end{figure}

Here we present the 
algorithm for constructing the inverted index for a given SFA,
as referred to in Section \ref{sec:index}.
The notations used are listed in Figure \ref{alg:indexing-notations}.

\begin{algorithm}[hbtp]
\label{alg:indexing-staccato-dyn}
$\forall e \in E~ with~ parent ~edges ~e' \in E$\\
\hspace{0.6cm} $\forall f_Q \in F_Q, ~AugSts_{par}(f) = \displaystyle\cup_{e'}^{}AugSts_{e'}(f_Q)$\\
\hspace{0.6cm} $~AugSts_e = RunDFA(AugSts_{par},e)$\\
$otherwise,$\\
\hspace{0.6cm} $\forall f_Q \in F_Q, ~AugSts_e(f_Q) = \phi$\\
\caption{The dynamic program for \name index construction}
\end{algorithm}

\begin{algorithm}
\For {$each ~string ~p_i ~on ~e, ~i=0 ~to~ k-1$} {
	$SO \leftarrow \{(0,0)\}$ //\{(State, Offset)\}\\
	\For {$each ~character ~c_j ~in ~p, ~j=0 ~to~ |p|-1$} {
		$NSO \leftarrow \phi$\\
		\For {$each ~ t \in SO$} {
			$Nxt \leftarrow \delta_Q(t.State,c_j)$\\
			\If {$Nxt \neq 0$} {
				$NSO \leftarrow NSO~\cup~{(Nxt, t.Offset)}$\\
				\If {$Nxt \in F_Q$} {
					$I(D(Nxt)) \leftarrow I(D(Nxt)) ~ \cup ~ (e,i,t.Offset)$\\
				}
			}
		}
		$SO \leftarrow NSO ~ \cup ~ \{(0,j+1)\}$\\
	}
	\For {$each ~ r ~ \in SO$} {
		\If {$r.State \neq 0$} {
			$NAugSts(r.State) \leftarrow NAugSts(r.State) \cup \{(e,i,r.Offset)\}$\\
		}
	}
	%\STATE //Run DFA from parents' augmented states
	\For {$each ~ d ~ \in AugSts$} {
		$Cur \leftarrow d.State$\\
		\For {$each ~character ~c_j ~in ~p, ~j=0 ~to~ |p|-1$} {
			$Nxt \leftarrow \delta(t.State,c_j)$\\
			\uIf {$Nxt \neq 0$} {
				$Cur = Nxt$\\
				\If {$Nxt \in F_Q$} {
					\For {$each ~ l ~ \in d.PostingSet$} {
						%$\textbf{Emit}~(D(Nxt), l)$\\
						$I(D(Nxt)) \leftarrow I(D(Nxt)) ~ \cup ~ \{l\}$\\
					}
				}
			}
			\Else {
				$break$ //DFA `dies' reading string\\
			}
			\If {$j = |p|-1$} {
				\For {$each ~ l ~ \in d.PostingSet$} {
					$NAugSts(Cur) \leftarrow NAugSts(Cur) \cup \{l\}$\\
				}
			}
		}
	}
}
\caption{\textit{RunDFA(AugSts,e)}}
\label{alg:indexing-staccato}
\end{algorithm}

The construction, presented in Algorithms \ref{alg:indexing-staccato-dyn}
and \ref{alg:indexing-staccato}, is similar to automata composition.
The dictionary of terms is first compressed into a trie-automaton \cite{hopcroft}
with multiple final states, each corresponding to a term. Then, we walk through the 
data SFA (using a dynamic program on the SFA's graph) and obtain the locations (postings) where any
dictionary term starts. A key thing to note here is that terms can straddle
mutiple SFA edges, which needs to be tracked. We pass information 
about such multi-edge terms through sets of `augmented states', which store pairs of the query 
DFA's state and possible postings. When the DFA reaches a final state (i.e., a term has 
been seen), the corresponding postings are added to the index.

\section{Implementation Details}
Each line of a document corresponds to one transducer, which is stored
as such in the FullSFA approach. 
$k$-MAP stores a ranked list of strings for each line
after inference on the transducer. In \name, each line corresponds to a
graph of chunks, where each chunk is a ranked list of strings. 
These data are stored inside the RDBMS with a relational
schema, shown in Table 5. There is one master table per
dataset, which contains the auxiliary information like document name,
line number, etc., and there are separate data tables for each
approach.

\begin{table}[hbtp]
\centering
\begin{tabular}{|c|c|p{2.5cm}|p{2.8cm}|l|}
\hline
\multirow{2}{*}{Approach} & \multirow{2}{*}{Table Name} & \multicolumn{2}{c|}{Attributes} & \multirow{2}{*}{Primary Key}\\
\cline{3-4}
& & Name & Type & \\
\hline
\hline
\multirow{3}{*}{-} & \multirow{3}{*}{MasterData} & DataKey & INTEGER & \multirow{3}{*}{\parbox{4cm}{DataKey}}\\
\cline{3-4}
 & & DocName & VARCHAR(50) & \\
\cline{3-4}
 & & SFANum & INTEGER & \\
\hline
\hline
\multirow{3}{*}{$k$-MAP} & \multirow{3}{*}{kMAPData} &  DataKey & INTEGER & \multirow{3}{*}{\parbox{4cm}{DataKey, LineNum}}\\
\cline{3-4}
 & & LineNum & INTEGER & \\
\cline{3-4}
& & Data & TEXT & \\
\cline{3-4}
 & & LogProb & FLOAT8 & \\
\hline
\hline
\multirow{2}{*}{FullSFA} & \multirow{2}{*}{FullSFAData} & DataKey & INTEGER & \multirow{2}{*}{\parbox{4cm}{DataKey}}\\
\cline{3-4}
& & SFABlob & OID &\\
\hline
\hline
\multirow{7}{*}{\name} & \multirow{5}{*}{StaccatoData} & DataKey & INTEGER & \multirow{5}{*}{\parbox{4cm}{DataKey, ChunkNum, \\LineNum}}\\
\cline{3-4}
 & & ChunkNum & INTEGER & \\
\cline{3-4}
 & & LineNum & INTEGER & \\
\cline{3-4}
 & & Data & TEXT & \\
\cline{3-4}
 & & LogProb & FLOAT8 & \\
\cline{2-5}
 & \multirow{2}{*}{StaccatoGraph} & DataKey & INTEGER & \multirow{2}{*}{\parbox{3cm}{DataKey}}\\ 
\cline{3-4}
& & GraphBlob & OID &\\
\hline
\hline
\multirow{2}{*}{-} & \multirow{2}{*}{GroundTruth} & DataKey & INTEGER & \multirow{2}{*}{\parbox{4cm}{DataKey}}\\
\cline{3-4}
 & & Data & TEXT & \\
\hline
\end{tabular}
\label{tab:schema}
\caption{Relational schema for storing SFA data}
\end{table}

\pagebreak
\section{Extended Experiments}

We now present more experimental results relating to runtimes and
answer quality for the filescans, as well as some aspects of the inverted indexing.
The queries we use are listed in Table \ref{exp:queries}, along 
with the number of ground truth answers for each on their respective datasets.

\begin{table}[hbtp]
\centering
\begin{tabular}{|c|c|c|c|}
\hline
Dataset & S.No. & Query & \# in Truth\\
\hline
\hline
\multirow{7}{*}{CA} & 1 & Attorney & 28	 \\
\cline{2-4}
 & 2 & Commission & 128	 \\
\cline{2-4}
 & 3 & employment &  73	 \\
\cline{2-4}
 & 4 & President  & 14	 \\
\cline{2-4}
 & 5 & United States  & 52	 \\
\cline{2-4}
 & 6 & Public Law $(8|9)\backslash d$ & 55	 \\
\cline{2-4}
 & 7 & U.S.C. $2\backslash d\backslash d\backslash d$  & 25	 \\
\hline
\hline
\multirow{7}{*}{DB} & 1 & accuracy  & 65	 \\
\cline{2-4}
 & 2 & confidence & 36	 \\
\cline{2-4}
& 3 & database & 43	 \\
\cline{2-4}
& 4 & lineage  & 83	 \\
\cline{2-4}
& 5 & Trio & 68	 \\
\cline{2-4}
& 6 & Sec$(\backslash x)*d$  & 33	 \\
\cline{2-4}
& 7 & $\backslash x\backslash x\backslash x\backslash d\backslash d$  & 47 \\
\hline
\hline
\multirow{7}{*}{LT} & 1 & Brinkmann  & 	92 \\
\cline{2-4}
& 2 & Hitler  & 12 \\
\cline{2-4}
& 3 & Jonathan & 18 \\
\cline{2-4}
& 4 & Kerouac  & 21	 \\
\cline{2-4}
& 5 & Third Reich  & 7	 \\
\cline{2-4}
& 6 & $19\backslash d\backslash d,~\backslash d\backslash d$ & 	32 \\
\cline{2-4}
& 7 & $spontan(\backslash x)*$ & 99 \\
\hline
\end{tabular}
\caption{Queries and ground truth numbers
}
\label{exp:queries}
\end{table}

\subsection{Recall and Runtime}
Table \ref{exp:precision-recall} presents the precision and recall results of the queries, 
while Table \ref{exp:runtimes} presents the respective runtime results.
The values of the parameters are $m=40,~k=50~ and ~NumAns=100$.
As in Section \ref{sec:experiments}, here too we see that \name lies between $k$-MAP and 
FullSFA on both recall and runtime. The precision too exhibits a similar trend. Again, the FullSFA 
approach is upto three orders of magnitude slower than MAP but achieves perfect recall on most queries, 
though precision is lower. Interestingly, on some queries in the DB dataset (e.g., DB3 and DB6), \name recall 
is close to 1.0 while the runtime is about two orders of magnitude lower than FullSFA.
Another thing to note is that the recall increase for $k$-MAP and \name (over MAP) is more pronounced in 
DB and LT than CA. We can also see that keyword queries can have lower recall
than some regex queries (e.g., LT3 and DB2).

\begin{table}[hbtp]
\centering
\begin{tabular}{|c||c|c|c||c|}
\hline
\multirow{2}{*}{Query} & \multicolumn{4}{c|}{Approach}\\
\cline{2-5}
& MAP & $k$-MAP & FullSFA & \name\\
\hline
\hline
CA1 & 1.00/0.93	& 1.00/0.93	& 0.28/1.00	& 0.87/0.96	\\
\hline
CA2 &  1.00/0.78	&  1.00/0.78	& 1.00/0.78	& 1.00/0.78	\\
\hline
CA3& 1.00/0.90	& 1.00/0.90	& 0.73/1.00	& 0.97/0.93 	\\
\hline
CA4&  1.00/0.79	& 1.00/0.79	& 0.14/1.00	& 0.85/0.79	 \\
\hline
CA5& 1.00/0.77	& 1.00/0.79	& 0.52/1.00 	& 1.00/0.88	\\
\hline
CA6&  1.00/0.87	&  1.00/0.96 & 0.55/1.00 & 1.00/0.98  \\
\hline
CA7& 1.00/0.28	&  1.00/0.52  & 0.25/1.00	& 0.50/0.80	\\
\hline
\hline
DB1 & 1.00/0.58	&  0.98/0.93  & 0.65/1.00	& 0.95/0.97	\\
\hline
DB2 & 0.00/0.00	& 0.87/0.19 & 0.36/1.00		& 0.90/0.53	\\
\hline
DB3 & 0.85/0.67	&  0.87/0.79  & 0.43/1.00	& 0.90/1.00	\\
\hline
DB4 & 0.97/0.91	& 0.96/0.93  & 0.82/0.99	& 0.85/0.95	\\
\hline
DB5 &   0.93/0.75	& 0.90/0.95 	& 0.67/0.99	& 0.79/0.96	\\
\hline
DB6 & 0.96/0.76	& 0.96/0.81  & 0.33/1.00	& 0.40/0.96	\\
\hline
DB7 &    0.91/0.85	& 0.73/0.89  & 0.44/0.94	& 0.42/0.89	\\
\hline
\hline
LT1 & 0.96/0.87 	& 0.96/0.90	& 0.92/1.00	& 0.94/0.91	\\
\hline
LT2 &    1.00/0.92 	& 1.00/1.00	& 0.12/1.00	& 0.12/1.00	\\
\hline
LT3 & 1.00/0.11	& 1.00/0.17 	& 0.18/1.00	& 0.94/0.83	\\
\hline
LT4 &   0.81/0.62 	& 0.86/0.90	& 0.21/1.00	& 0.74/0.95	\\
\hline
LT5 &  1.00/0.29 	& 1.00/1.00	& 0.07/1.00	& 1.00/1.00	\\
\hline
LT6 & 0.77/0.65	& 0.76/0.67	& 0.31/0.97	& 0.26/0.81	\\ 
\hline
LT7 &    0.84/0.88	& 0.83/0.88  & 0.83/0.88	& 0.83/0.88	\\
\hline
\end{tabular}
\caption{Precision and recall results
}
\label{exp:precision-recall}
\end{table}

\begin{table}[hbtp]
\centering
\begin{tabular}{|c||c|c|c||c|}
\hline
\multirow{2}{*}{Query} & \multicolumn{4}{c|}{Approach}\\
\cline{2-5}
& MAP & $k$-MAP & FullSFA & \name\\
\hline
\hline
CA1 & 0.17	& 0.82	& 81.54	& 4.38	\\
\hline
CA2 &  0.17	& 0.96	& 91.84	& 4.93	\\ 
\hline
CA3 & 0.17	& 0.96	& 91.85	& 4.94	\\
\hline
CA4 & 0.17	& 0.89	& 86.72	& 4.63	\\
\hline
CA5 & 0.18	& 1.16	& 106.17	& 5.97	\\
\hline
CA6 & 0.18	& 1.17	& 125.63 & 5.98	\\ 
\hline
CA7 & 0.18	& 1.05	& 150.35	& 5.40	\\
\hline
\hline
DB1 & 0.07	& 0.44	& 56.42	& 1.61	\\
\hline
DB2 & 0.07	& 0.51	& 62.89	& 1.81	\\
\hline
DB3 & 0.07	& 0.43	& 54.92	& 1.59	\\
\hline
DB4 & 0.07	& 0.40	& 51.45	& 1.48	\\
\hline
DB5 & 0.07	& 0.42	& 40.72	& 1.21	\\
\hline
DB6 & 0.07	& 0.35	& 619.31	& 1.39	\\
\hline
DB7 & 0.07	& 0.31	& 1738.78	& 1.37	\\
\hline
\hline
LT1 & 0.14	& 0.73	& 83.78	& 3.27	\\
\hline
LT2 & 0.13	& 0.59	& 69.68	& 2.72	\\
\hline
LT3 & 0.14	& 0.71	& 79.76	& 3.10	\\
\hline
LT4 & 0.14	& 0.65	& 74.58	& 2.90	\\
\hline
LT5 & 0.14	& 0.85	& 93.35	& 3.72	\\
\hline
LT6 & 0.14	& 1.02	& 155.45	& 4.52	\\ 
\hline
LT7 & 0.15	& 1.00	& 887.19	& 4.23	\\
\hline
\end{tabular}
\caption{Runtime results. Runtimes are in seconds.
%The parameter setting for \name is: $m=40,~k=50~ and ~NumAns=100$. 
}
\label{exp:runtimes}
\end{table}

\pagebreak

\subsection{Precision and F-1 Score}
Though our focus is on recall-sensitive applications, we also study how the 
precision is affected when we vary $m$ and $k$. For the same queries 
and parameter setting as in Figure \ref{recall-runtime-vs-k}, we plot the precision and 
F-1 score of the answers obtained. Figure \ref{exp:precision-f1} shows the results.

\begin{figure}[hbtp]
\centering
\includegraphics[width=4in]{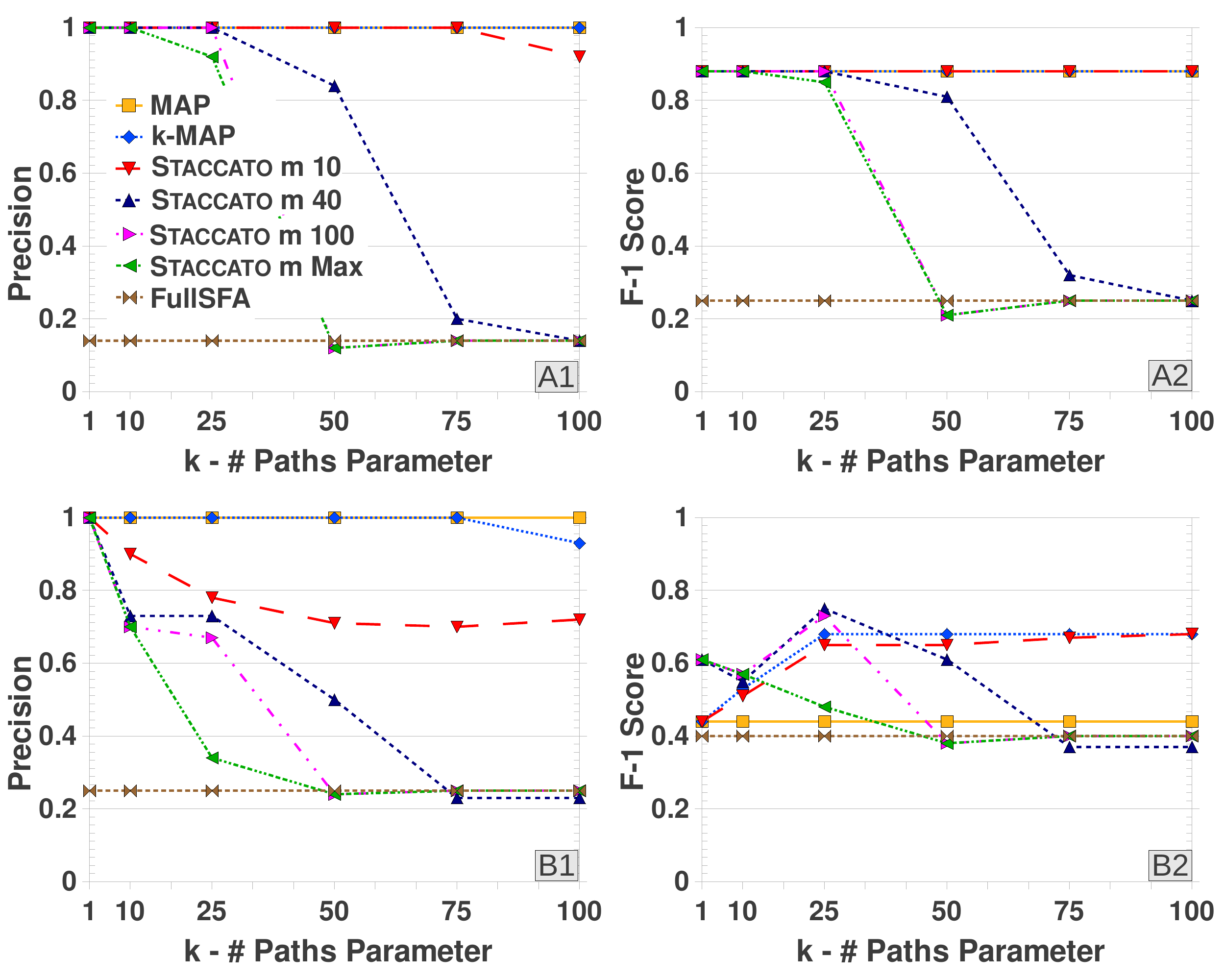}\\
\caption{Precision and F-1 Score variations with $k$ on two queries: (A) `$President$', and 
(B) `$U.S.C.~2 \backslash d \backslash d \backslash d$'. $NumAns$ is set to 100.}
\label{exp:precision-f1}
\end{figure}

As mentioned before, $k$-MAP precision is high, since it returns only a few answers, almost all
of which are correct. On the other hand, FullSFA precision is lowest, since it returns many
incorrect answers (along with most of the correct ones). Again, \name falls in between, with 
the precision being high (close to $k$-MAP) for lower values of $m$ and $k$, and gradually 
drops as we increase $m$ and $k$. Also, the precision drops faster for higher values of $m$. 
It should be noted that the precision needn't drop monotonically, since additional correct
answers might be obtained at higher $m$ and $k$, boosting both the recall and precision. 
For completeness sake, the F-1 score variation is also presented in Figure \ref{exp:precision-f1}.
Interestingly, for the regex query, the F-1 score of both $k$-MAP and FullSFA are lower than 
that of \name, the former due to its lower recall, and the latter due to its lower precision. 

\subsection{Sensitivity to NumAns}

As we mentioned in Section \ref{sec:experiments}, the quality of the answers obtained is sensitive
to the $NumAns$ parameter. If it is set too low (lower than the number of ground truth answers),
the recall is likely to be low. On the other hand, if it is set too high, the recall will 
increase, but the precision might suffer. Thus, we perform a sensitivity analysis on $NumAns$
for the recall, precision and F-1 score obtained. The same queries as in 
Figure \ref{recall-runtime-vs-k} are used and the parameter setting is $m=40,~k=75$. Figure
\ref{exp:numans} shows the results.

\begin{figure}[hbtp]
\centering
\includegraphics[width=4.5in]{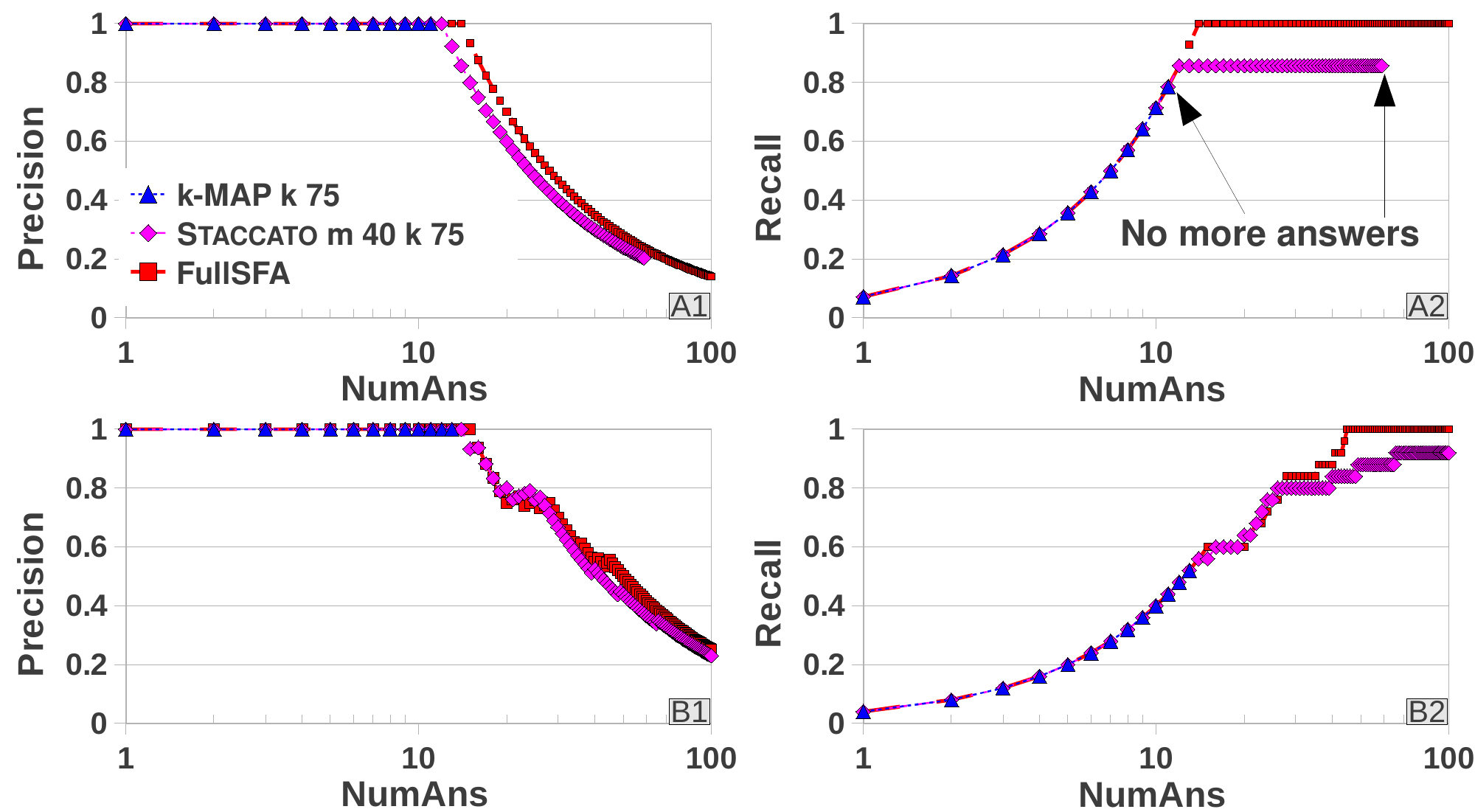}\\
\caption{Sensitivity of Precision and Recall to $NumAns$ on two queries: (A) $President$, and 
(B) $U.S.C.~2 \backslash d \backslash d \backslash d$. The x-axes are in logscale.}
\label{exp:numans}
\end{figure}

\begin{figure}[hbtp]
\centering
\includegraphics[width=5.5in]{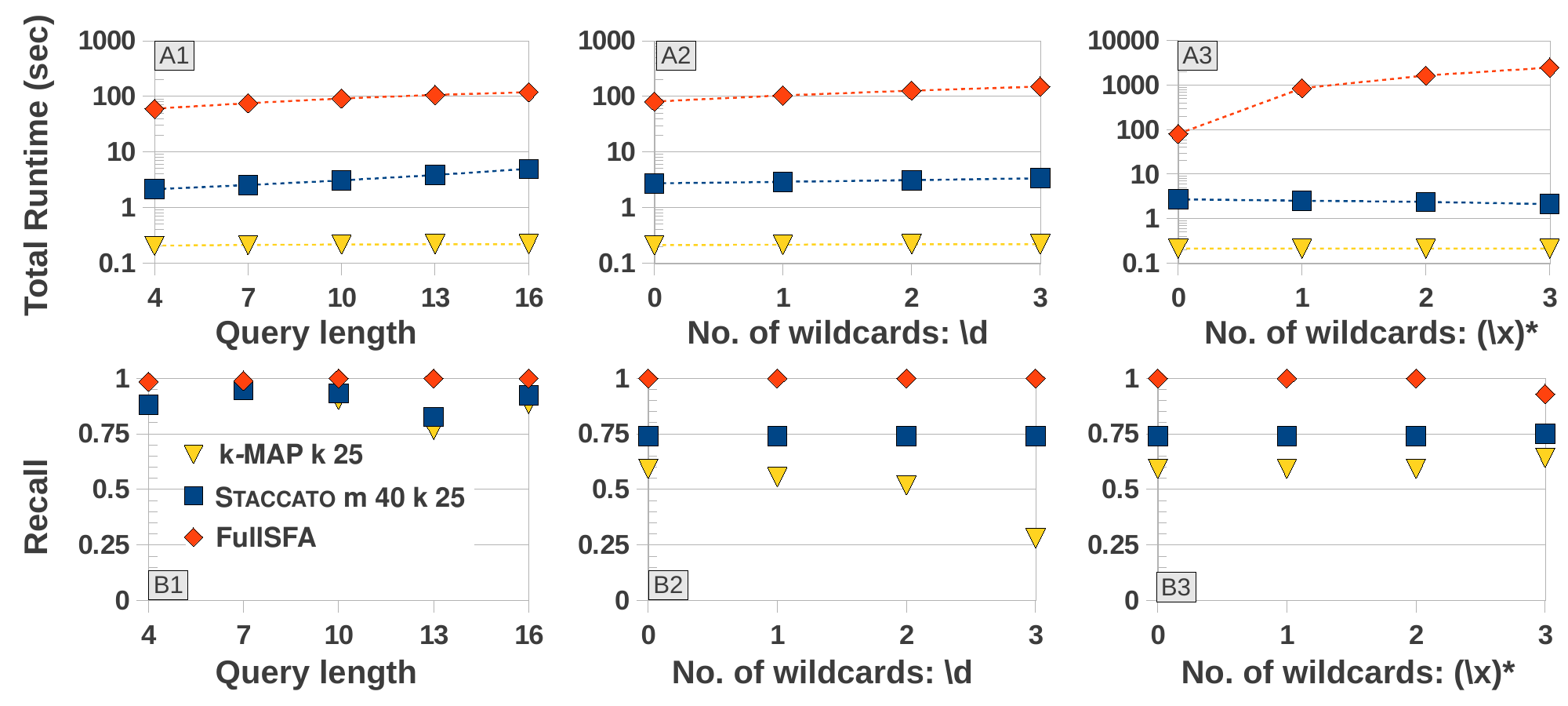}\\
\caption{Impact of Query Length and Complexity: (A1, B1) correspond to the 
keyword queries, (A2, B2) correspond to the simple wildcard queries, while 
the last pair correspond to the complex wildcard queries. $NumAns$ is set to 100. 
Runtimes are in logscale.}
\label{exp:qlenexpts}
\end{figure}

\begin{figure}[hbtp]
\centering
\includegraphics[width=3in]{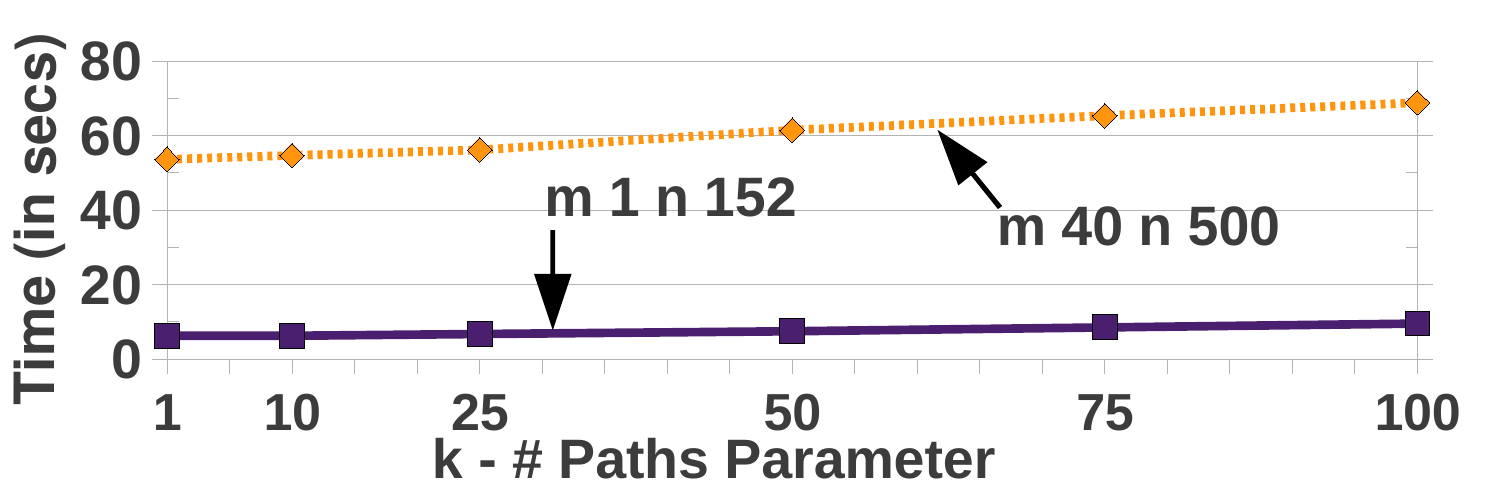}
\caption{Sensitivity of \name construction time to the parameter $k$
} 
\label{exp:stacconstr-k}
\end{figure}

As Figures \ref{exp:numans} (A1) and (A2) show, the precision initially remains high (here, at 1) when 
$NumAns$ is low. This is because the highest probability answers (that appear on top) are likely to be correct.
As we increase $NumAns$, we get more correct answers and thus recall increases constantly. Once we near
the number of ground truth answers, the recall starts to flatten, while the precision starts to drop. For $k$-MAP, beyond
a value of $NumAns$, no more answers are returned since no matches exist. On the other hand, in FullSFA, almost
all SFAs match almost all queries, and so we keep getting answers. Moreover, since FullSFA uses the full probabilities,
the increase is relatively smooth. \name for the given parameters gives more answers
than $k$-MAP, and achieves higher recall but falls short of FullSFA. The recall for the regex query in \name recall 
keeps increasing after being relatively flat many times. Thus, \name can still achieve high recall, though
at a lower precision than FullSFA.

\subsection{Effect of Query Length and Complexity}

We now study the impact of they query length and complexity on the runtime and recall obtained.
For this study, we use three sets of queries. The first set consists of keyword queries of 
increasing length.
The second set consists of regular expression queries with an increasing number of simple wildcards, e.g., 
`$U.S.C.~2\backslash d\backslash d$' (where $\backslash d$ is any digit). 
The third set too consists of regular expression queries, but with the more complex Kleene star as wildcards, e.g.,
`$U(\backslash x)*S(\backslash x)*C.~2$' (where $\backslash x$ is any character). 
Figure \ref{exp:qlenexpts} shows the runtime and recall results for these queries.
\eat{(`$case$', `$section$', `$employment$', `$United~States$' and `$Attorney~General$')
(`$U.S.C. 2$', `$U.S.C.~2\backslash d$', `$U.S.C.~2\backslash d\backslash d$' and `$U.S.C.~2\backslash d\backslash d\backslash d$)
(`$U.S.C.~2$', `$U(\backslash x)*S.C.~2$', `$U(\backslash x)*S(\backslash x)*C.~2$' and `$U(\backslash x)*S(\backslash x)*C(\backslash x)*2$')}

As expected, the plots show that the runtime increases polynomially, but slowly with query length in all the approaches. 
However, the increase is more pronounced for FullSFA with complex wildcards (Figure \ref{exp:qlenexpts}:A3)
since the composition based query processing produces much larger intermediate results. It can also be seen that there 
is no definite trend in the obtained recall for (Figure \ref{exp:qlenexpts}:A1), since a longer query can have better recall
than a shorter one.

%\subsection{Runtime Breakdown}
%\subsection{Inverted Index Usage}
\subsection{Staccato Construction}
We now present the sensitivity of the \name construction times to the 
parameter $k$. Figure \ref{exp:stacconstr-k} shows the results.

The plot shows that the runtimes increase linearly with $k$.
However, as mentioned in Section \ref{sec:experiments}, this linearity is
not guaranteed since the chunk structure obtained may not be the same across different 
values of $k$, for a fixed SFA and $m$.

\subsection{Index Construction Time}
Here, we discuss the runtimes for the \name index construction, 
which is a two-phase process. First, we obtain the postings independently 
for each SFA, and then unify all postings into the index. 
We pick a few SFAs and run the indexing in a controlled setting. 

\begin{figure}[hbtp]
\centering
\includegraphics[width=4in]{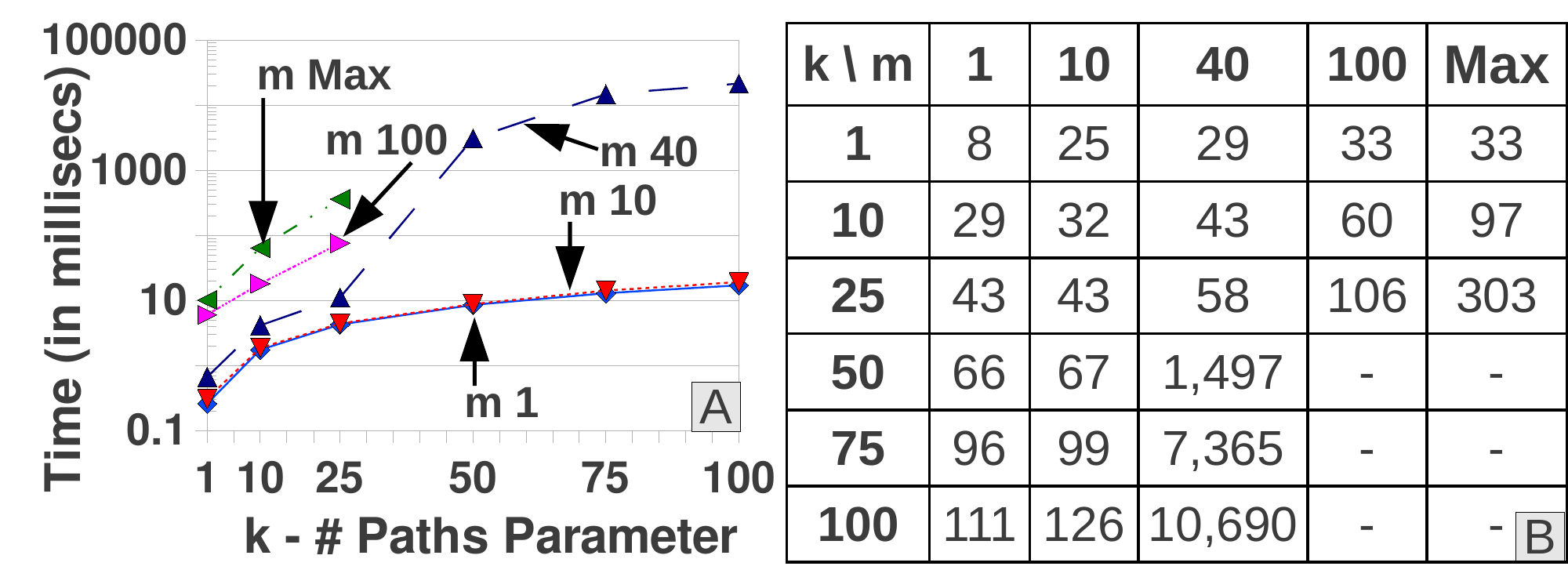}
\caption{(A) \name index construction times. Note the logscale on the y-axis. (B)  
Bulk index load times (in seconds) for the index tables of the LT dataset.
} 
\label{ftip-constrtimes}
\end{figure}

Figure \ref{ftip-constrtimes} shows the sensitivity of the construction times for $m$ and $k$ for a single SFA, and also tabulates
the bulk index load times for an entire dataset (LT). Firstly, we 
can see these runtimes are mostly practical. Also, we can see a linear trend in $k$, with a non-linearity showing up at $m =40, k=50$.
we found that this was because the data in this parameter space had  
many single-character wide blocks, leading to the presence of more terms, and blowing up the size of the index. This causes two 
effects - the number of postings per SFA goes up by upto three orders of magnitude, and the selectivity of most terms across the 
dataset shoots up too, as was seen in Figure \ref{indexed-runtimes} (A).
%To understand this better, 
%we checked the total number of postings (across terms) produced (Figure \ref{ftip-constrtimes}). It shows
%a linear dependence between the runtime and index size for most of the parameter settings. However, for $M=40$, there is a 
%non-linear effect as $K$ goes from 25 to 50, wherein many new terms start to appear.

Since running the indexing on all SFAs is also easily parallelizable,
we again used Condor \cite{condor}. Overall, the index construction on
all the data, for the above parameters took about 3 hours. After
obtaining the postings lists for all SFAs, we loaded them all into the
index table. These bulk load times are tabulated in Figure
\ref{ftip-constrtimes} (B). We can see that the load times are concomitant 
with the construction times due to the size of the index obtained.

\subsection{Index Utility and Size}

\begin{figure}[hbtp]
\centering
\includegraphics[width=4in]{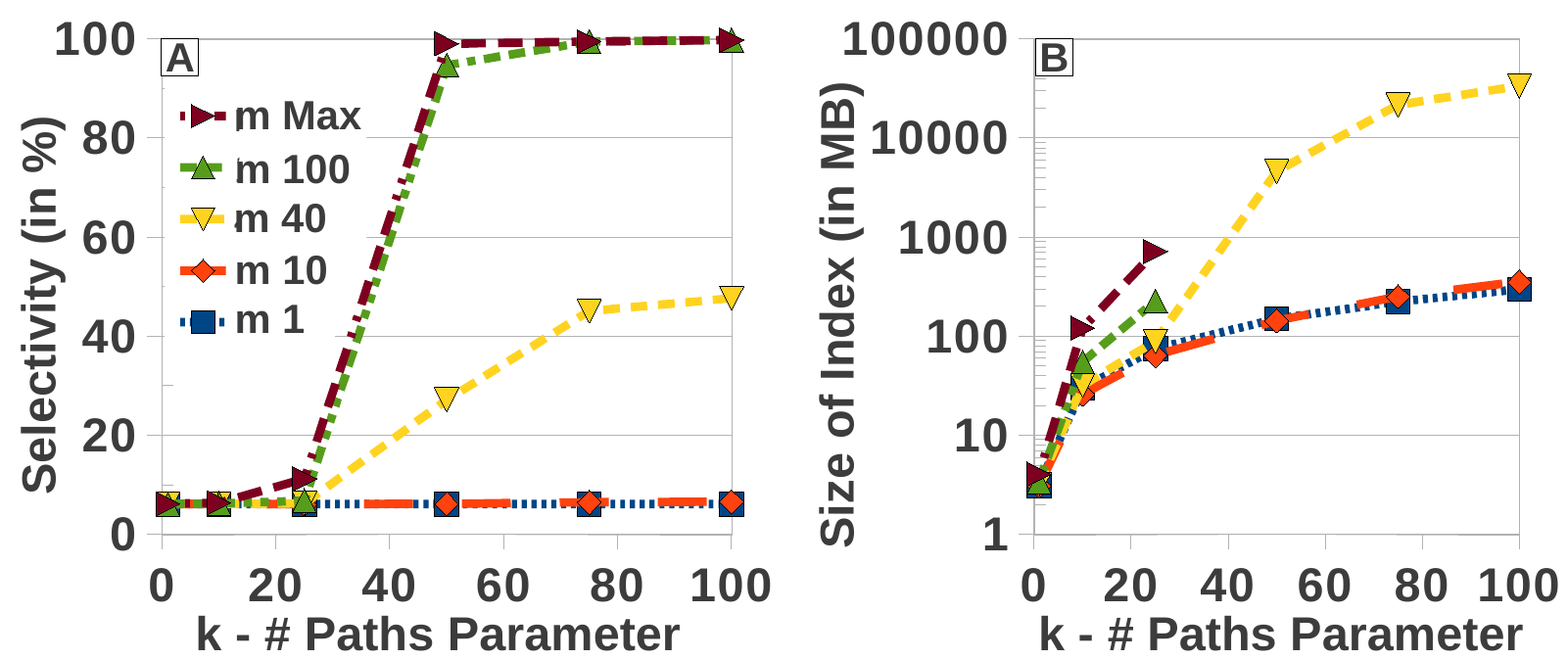}
\caption{(A) Selectivity (\%ge of SFAs) of the term `$public$' using the \name index, for various
values of $m$ and $k$.
(B) Size of the \name index. Note the logscale on the y-axis.
} 
\label{ftip-sizeselec}
\end{figure}

It was discussed in Section \ref{sec:experiments} that the inverted index becomes less `useful'
as $m$ and $k$ become higher. To justify that, we study the selectivity of a query that uses the 
index. Here, we define selectivity as the percentage of the SFAs in the dataset that match the 
query when using the index. Figure \ref{ftip-sizeselec} (A) shows the results for a query on the 
CA dataset. A complementary aspect of the utility of the index is its size. 
Figure \ref{ftip-sizeselec} (B) shows the size of the index over the data in \name.

Two interesting things can be seen from these plots. The query selectivity for lower values of 
$m$ and $k$ is relatively low, but for middle values of $m$ ($m = 40$), the term
starts to appear in many more SFAs as $k$ increases. For higher values of $m$ ($m=100,Max$), 
as $k$ increases, the the selectivity shoots up to nearly 100\%, which means that almost 
all SFAs in the dataset contain the term.
This implies that the index is no longer useful in the sense that almost the entire
dataset is returned as answers. We observed the same behavior across all queries and datasets.
The plot of the index sizes reflects this phenomenon. The size varies largely linearly as expected, but 
at $m=40,k=50$, it shoots up two orders of magnitude, similar to Figure \ref{ftip-constrtimes}. 
This size increase is largely because of the 
selectivity increase, i.e., many more entries appearing in the index.
The index construction for $m=100,Max$ for $k=50$, and above was skipped since the index sizes exceeded 
the available disk space (over 200 GB).
However, the selectivity can be easily computed after obtaining just the first posting, with no need to 
compute all the postings. The selectivity confirms that these parameter settings do not give useful indexes anyway.

\end{document}